\documentclass[submission,copyright]{eptcs}
\providecommand{\event}{VPT 2016} 

\usepackage{stmaryrd}
\usepackage{latexsym}
\usepackage{amsmath}
\usepackage{amssymb}
\usepackage{wasysym}
\usepackage{proof}
\usepackage{tikz}

\newcommand{\expr}[1]{#1{}}
\newcommand{\cont}[2]{#1{#2}}
\newcommand{\var}[2]{\mathit{#1}#2}
\newcommand{\abs}[3]{\lambda #1{.#2{#3}}}
\newcommand{\app}[3]{#1{~#2{#3}}}
\newcommand{\letexp}[4]{{\bf let} ~ #1{ =  #2{~ {\bf in} ~ #3{#4}}}}

\newcommand{\args}[3]{#1{\ldots#2{#3}}}
\newcommand{\fundef}[2]{#1 & = & \expr{#2}}

\newcommand{\cas}[6]{\begin{array}[t]{@{\hspace*{0mm}}l@{\hspace*{1mm}}l@{\hspace*{1mm}}c@{\hspace*{1mm}}l@{\hspace*{0mm}}} 
\multicolumn{4}{@{\hspace*{0mm}}l@{\hspace*{0mm}}}{{\bf case} ~ #1 ~ {\bf of}} \\
& #2{} & \rightarrow & #3{} \\
~~~ | & #4{} & \rightarrow & #5{#6}\end{array}}
\newcommand{\longcas}[8]{\begin{array}[t]{@{\hspace*{0mm}}l@{\hspace*{1mm}}l@{\hspace*{1mm}}c@{\hspace*{1mm}}l@{\hspace*{0mm}}} 
\multicolumn{4}{@{\hspace*{0mm}}l@{\hspace*{0mm}}}{{\bf case} ~ #1 ~ {\bf of}} \\
& #2{} & \rightarrow & #3{} \\
~~~ | & #4{} & \rightarrow & #5{} \\
~~~ | & #6{} & \rightarrow & #7{#8}\end{array}}

\newcommand{\smallcas}[4]{\begin{array}[t]{@{\hspace*{0mm}}l@{\hspace*{1mm}}l@{\hspace*{1mm}}c@{\hspace*{1mm}}l@{\hspace*{0mm}}}
\multicolumn{4}{@{\hspace*{0mm}}l@{\hspace*{0mm}}}{{\bf case} ~ #1 ~ {\bf of}} \\ ~~~~~ #2 ~ \rightarrow ~  #3{#4}\end{array}}
\newcommand{\casedots}[6]{{\bf case}~#1{~{\bf of}~#2{\rightarrow #3{~| 
\cdots|~#4{\rightarrow #5{#6}}}}}}
\newcommand{\where}[3]{\!\!\!\begin{array}[t]{@{\hspace*{0mm}}l@{\hspace*{1mm}}c@{\hspace*{1mm}}l@{\hspace*{0mm}}}\multicolumn{3}{@{\hspace*{0mm}}l@{\hspace*{0mm}}}{#1{}}\\\multicolumn{3}{@{\hspace*{0mm}}l@{\hspace*{0mm}}}{{\bf where}}\\#2{#3}\end{array}}

\newcommand{\Where}[6]{#1{} ~ {\bf where} ~ #2 = #3 \ldots #4 = #5{#6}}
\newcommand{\brackets}[2]{(#1{)#2}}

\newcommand{\prove}[5]{{\cal P}[\![#1{]\!]~#2{~#3{~#4{#5}}}}}
\newcommand{\generate}[6]{{\cal C}[\![#1{]\!]~#2{~#3{~#4{~#5{#6}}}}}}

\newcommand{\concat}[3]{#1{+\!\!+ #2{#3}}}

\newcommand{\wildcard}[0]{\underline{\hspace{2mm}}}
\newcommand{\ignore}[1]{}

\newenvironment{proof}{{\em Proof}.~}{\hfill $\Box$\\}

\newtheorem{theorem}{Theorem}[section]

\newtheorem{example}{Example}  
\newtheorem{property}{Property}  

\title{Generating Counterexamples for Model Checking by Transformation}

\author{G.W. Hamilton
\institute{School of Computing and Lero\\
Dublin City University\\
Ireland}
\email{hamilton@computing.dcu.ie}}

\def\titlerunning{Generating Counterexamples by Transformation}

\def\authorrunning{G.W. Hamilton}

\begin{document}

\maketitle

\begin{abstract}
Counterexamples explain why a desired temporal logic property fails to hold. The generation of counterexamples is considered 
to be one of the primary advantages of model checking as a verification technique. Furthermore, when model checking does 
succeed in verifying a property, there is typically no independently checkable witness that can be used as evidence for the 
verified property. Previously, we have shown how program transformation techniques can be used for the verification of both 
safety and liveness properties of reactive systems. However, no counterexamples or witnesses were generated using the 
described techniques. In this paper, we address this issue. In particular, we show how the program transformation technique 
{\em distillation} can be used to facilitate the construction of counterexamples and witnesses for temporal properties of reactive
systems. Example systems which are intended to model mutual exclusion are analysed using these techniques with respect to 
both safety (mutual exclusion) and liveness (non-starvation), with counterexamples being generated for those properties which
do not hold.
\end{abstract}

\section{Introduction}

Model checking is a well established technique originally developed for the verification of temporal properties of 
finite state systems \cite{CLARKE86}. In addition to telling the user whether the desired temporal property holds, 
it can also generate a {\em counterexample}, explaining the reason why this property failed. This is considered to be 
one of the major advantages of model checking when compared to other verification methods. Fold/unfold program 
transformation techniques have more recently been proposed as an approach to model checking. 
Many such techniques have been developed for logic programs (e.g. \cite{LEUSCHEL99,ROYCHOUDHURI00,FIORAVANTI01,PETTOROSSI09,SEKI11}). However, very few such techniques 
have been developed for functional programs (with the work of Lisitsa and Nemytykh \cite{LISITSA07,LISITSA08} using
supercompilation \cite{TURCHIN86} being a notable exception), and these deal only with safety properties. Unfortunately, 
none of these techniques generate counterexamples when the temporal property does not hold.

In previous work \cite{HAMILTON15}, we have shown how a fold/unfold program transformation technique can be used to 
facilitate the verification of both safety and liveness properties of reactive systems which have been specified using functional 
programs. These functional programs produce a {\em trace} of states as their output, and the temporal property specifies the 
constraints that all output traces from the program should satisfy. However, counterexamples and witnesses were not generated 
using this approach. In this paper, we address this shortcoming to show how our previous work can be extended to generate 
a counterexample trace when a temporal property does not hold, and a witness when it does. 

The program transformation technique which we use is our own {\em distillation} \cite{HAMILTON07A,HAMILTON12} 
which builds on top of positive supercompilation \cite{SORENSEN96}, but is much more powerful. Distillation is used to 
transform the programs defining reactive systems into a simplified form which makes them much easier to analyse. 
We then show how temporal properties for this simplified form can be verified, and extend this to generate counterexamples
and witnesses. The described techniques are applied to a number of example systems which are intended to model mutually 
exclusive access to a critical resource by two processes. When a specified temporal property does not hold, we show how our 
approach can be applied to generate a corresponding counterexample and when the property does hold we show our approach
can be applied to generate a corresponding witness.

The remainder of this paper is structured as follows. In Section 2, we introduce the functional language over which our
verification techniques are defined. In Section 3, we show how to specify reactive systems in our language, 
and give a number of example systems which are intended to model mutually exclusive access to a critical resource by two 
processes. In Section 4, we describe how to specify temporal properties for reactive systems defined in our language, and
specify both safety (mutual exclusion) and  liveness (non-starvation) for the example systems. In Section 5, we describe our 
technique for verifying temporal properties of reactive systems and apply this technique to the example systems to verify the 
previously specified temporal properties. In Section 6, we describe our technique for the generation of counterexamples and
witnesses, and apply this technique to the example systems. Section 7 concludes and considers related work.

\section{Language}
 \label{sec-language-definition}

In this section, we describe the syntax and semantics of the higher-order functional language which will be used 
throughout this paper. 

\subsection{Syntax}

The syntax of our language is given in Figure \ref{grammar}.
\begin{figure}[htb]
\begin{center}
\begin{tabular}{@{\hspace*{0mm}}l@{\hspace*{1mm}}r@{\hspace*{1mm}}l@{\hspace*{1mm}}l@{\hspace*{0mm}}}
$\expr{\var{e}}$ & ::= & $\expr{\var{x}}$ & Variable \\
& $|$ & $\expr{\app{\var{c}}{\args{\var{e_1}}{\var{e_k}}}}$ & Constructor Application \\
& $|$ & $\expr{\abs{\var{x}}{\var{e}}}$ & $\lambda$-Abstraction \\
& $|$ & $\expr{\var{f}}$ & Function Call \\
& $|$ & $\expr{\app{\var{e_0}}{\var{e_1}}}$ & Application \\
& $|$ & $\expr{\casedots{\var{e_0}}{\var{p_1}}{\var{e_1}}{\var{p_k}}{\var{e_k}}}$ & Case Expression \\ 
& $|$ & $\expr{\letexp{\var{x}}{\var{e_0}}{\var{e_1}}}$ & Let Expression \\
& $|$ & $\expr{\Where{\var{e_0}}{\var{f_1}}{\var{e_1}}{\var{f_n}}{\var{e_n}}}$ & Local Function Definitions \\
\\
$\expr{\var{p}}$ & ::= & $\expr{\app{\var{c}}{\args{\var{x_1}}{\var{x_k}}}}$ & Pattern
\end{tabular}
\end{center}
\caption{Language Grammar}
\label{grammar}
\end{figure} \\
A program is an expression which can be a variable, constructor application, $\lambda$-abstraction, 
function call, application, {\bf case}, {\bf let} or {\bf where}. 
Variables introduced by $\lambda$-abstractions, {\bf let} expressions and {\bf case} patterns are {\em bound}; all other variables 
are {\em free}. An expression which contains no free variables is said to be {\em closed}.

Each constructor has a fixed arity; for example $\expr{\var{Nil}}$ has arity 0
and $\expr{\var{Cons}}$ has arity 2.  In an expression $\expr{\app{\var{c}}{\args{\var{e_{1}}}{\var{e_{n}}}}}$,
$n$ must equal the arity of $c$. 
The patterns in {\bf case} expressions may not be nested.  No variable may appear more than once within a pattern
and the same constructor cannot appear within more than one pattern.
We assume that the patterns in a {\bf case} expression are exhaustive; we also allow a wildcard pattern
$\wildcard$ which always matches if none of the earlier patterns match. Types are defined using algebraic data types, and it is 
assumed that programs are well-typed. Erroneous terms such as $\expr{\casedots{\brackets{\abs{x}{e}}}{\var{p_1}}{\var{e_1}}{\var{p_k}}{\var{e_k}}}$ and $\expr{\app{\brackets{\app{\var{c}}{\args{\var{e_1}}{\var{e_n}}}}}{\var{e}}}$ where $c$ is of arity $n$ cannot therefore occur.

\subsection{Semantics}

The call-by-name operational semantics of our language is standard: we define an evaluation relation $\Downarrow$ between 
closed expressions and {\em values}, where values are expressions in {\em weak head normal form} (i.e. constructor applications or $\lambda$-abstractions). 
We define a one-step reduction relation $\overset{r}{\leadsto}$ inductively as shown in Figure \ref{reduction}, where the reduction $r$ can be $f$ (unfolding of function $f$), $c$ (elimination of constructor $c$) or $\beta$ ($\beta$-substitution). 
\begin{figure}[htb]
\begin{center}
\begin{tabular}{c@{\hspace*{1cm}}c}
$((\lambda x.e_0)~e_1) \overset{\beta}{\leadsto} (e_0\{x \mapsto e_1\})$ & $(\expr{\letexp{\var{x}}{\var{e_0}}{\var{e_1}}}) \overset{\beta}{\leadsto} (e_1\{x \mapsto e_0\})$ \\
\\
$\infer{f \overset{f}{\leadsto} e}{f=e}$ & $\infer{(e_0~e_1) \overset{r}{\leadsto} (e_0'~e_1)}{e_0 \overset{r}{\leadsto} e_0'}$  \\
\\
\multicolumn{2}{c}{$\infer{(\mathbf{case}~(c~e_1 \ldots e_n)~\mathbf{of}~p_1:e_1' | \ldots | p_k:e_k') \overset{c}{\leadsto} (e_i\{x_1 \mapsto e_1,\ldots,x_n \mapsto e_n\})}{p_i=c~x_1 \ldots x_n}$} \\
\\
\multicolumn{2}{c}{$\infer{(\mathbf{case}~e_0~\mathbf{of}~p_1:e_1 | \ldots p_k:e_k) \overset{r}{\leadsto} (\mathbf{case}~e_0'~\mathbf{of}~p_1:e_1 | \ldots p_k:e_k)}{e_0 \overset{r}{\leadsto} e_0'}$}
\end{tabular} 
\end{center}
\caption{One-Step Reduction Relation}
\label{reduction}
\end{figure} \\
We use the notation $e\leadsto$ if the expression $e$ reduces, $e\!\Uparrow$ if $e$ diverges, $e\!\Downarrow$ if $e$ converges and 
$e\!\Downarrow\!v$ if $e$ evaluates to the value $v$. These are defined as follows, where $\overset{r*}{\leadsto}$ denotes the reflexive transitive closure of $\overset{r}{\leadsto}$:
\begin{center}
\begin{tabular}{l@{\hspace{0.5cm}}l}
$e\leadsto$, iff $\exists e'.e \overset{r}{\leadsto} e'$ & $e\!\Downarrow$, iff $\exists v.e\!\Downarrow\!v$ \\
$e\!\Downarrow\!v$, iff $e \overset{r*}{\leadsto}v \wedge \neg(v\leadsto)$ & $e\!\Uparrow$, iff $\forall e'.e \overset{r*}{\leadsto}e' \Rightarrow e'\leadsto$
\end{tabular}
\end{center}
\section{Specifying Reactive Systems}

In this section, we show how to specify reactive systems in our programming language.
While reactive systems are usually specified using {\em labelled transitions systems} (LTSs), our
specifications can be trivially derived from these.
Reactive systems have to react to a series of {\em external events} by updating their {\em states}.
In order to facilitate this, we make use of a {\em list} datatype, which is defined as follows for the element type $a$:
$$List~a ::= Nil~|~Cons~a~(List~a)$$
We use $[]$ as a shorthand for $Nil$, and $[s_1, \ldots, s_n]$ as a shorthand for a list containing the elements $s_1 \ldots s_n$.
We also use $+\!\!+$ to represent list concatenation.
Our programs will map a (potentially infinite) input list of external events and an initial state to a (potentially infinite) output list of 
{\em observable states} (a {\em trace}), which gives the values of a subset of state variables whose properties can be verified.

In this paper, we wish to analyse a number of systems which are intended to implement mutually exclusive
access to a critical resource for two processes. In all of these systems, the external events belong to the following datatype: 
$$Event ::= Request_1~|~Request_2~|~Take_1~|~Take_2~|~Release_1~|~Release_2$$
Each of the two processes can therefore request access to the critical resource, and take and release
this resource. Observable states in all of our example systems belong to the following datatype:
$$State ::= ObsState~ProcState~ProcState$$
$$ProcState ::= T~|~W~|~U$$
Each process can therefore be thinking ($T$), waiting for the critical resource ($W$) or using the critical resource ($U$). 

Each of our example systems is transformed into a simplified form as previously shown in \cite{HAMILTON15} using distillation 
\cite{HAMILTON07A,HAMILTON12}, a powerful program transformation technique which builds on top of the supercompilation 
transformation \cite{TURCHIN86,SORENSEN96}. Due to the nature of the programs modelling reactive systems, in which the input 
is an external event list, and the output is a list of observable states, the programs resulting from this transformation take the form 
$e^{\emptyset}$, where $e^{\rho}$ is defined as shown in Figure \ref{simplified} where the {\bf let} variables are added to the set 
$\rho$, and will not be used as {\bf case} selectors.
\begin{center}
\begin{figure}[htbp]
\begin{center}
\begin{tabular}{lrl}
$\expr{\var{e^{\rho}}}$ & ::= & $\expr{\app{\app{\var{Cons}}{\var{e_0^{\rho}}}}{\var{e_1^{\rho}}}}$ \\
& $|$ & $\expr{\app{\var{f}}{\args{\var{x_1}}{\var{x_n}}}}$ \\
& $|$ & $\expr{\casedots{\var{x}}{\var{p_1}}{\var{e_1^{\rho}}}{\var{p_k}}{\var{e_n^{\rho}}}}$, where $x \notin \rho$ \\ 
& $|$ & $\expr{\app{\var{x}}{\args{\var{e_1^{\rho}}}{\var{e_n^{\rho}}}}}$, where $x \in \rho$ \\
& $|$ & $\expr{\letexp{\var{x}}{\abs{\args{\var{x_1}}{\var{x_n}}}{\var{e_0^{\rho}}}}{\var{e_1^{(\rho \cup \{x\})}}}}$ \\
& $|$ & $\expr{\Where{\var{e_0^{\rho}}}{\var{f_1}}{\abs{\args{\var{x_{1_1}}}{\var{x_{1_k}}}}{\var{e_1^{\rho}}}}{\var{f_n}}{\abs{\args{\var{x_{n_1}}}{\var{x_{n_k}}}}{\var{e_n^{\rho}}}}}$
\end{tabular} 
\end{center}
\label{simplified}
\caption{Simplified Form Resulting From Distillation}
\end{figure} 
\end{center}
The crucial syntactic property of this simplified form is that all functions must be tail recursive; this is what allows the resulting
programs to be verified more easily.
In all of the following examples, the variable $es$ represents the external event list.
\begin{example}
\normalfont{In the first example shown in Figure \ref{example1}, each process can request access to the critical resource if 
it is thinking and the other process is not using it, take the critical resource if it is waiting for it, and release the critical resource 
if it is using it. The LTS representation of this program is shown in Figure \ref{example1lts} (for ease of presentation of this and
subsequent LTSs, transitions back into the same state have been omitted).}
\end{example}
\begin{figure}[htbp]
\hspace*{2cm}
\begin{tabular}{l}
$\expr{\where{\app{\app{\var{Cons}}{\var{(ObsState~T~T)}}}{\brackets{\app{\var{f_1}}{\var{es}}}}}{
\fundef{\var{f_1}}{\abs{\var{es}}{\smallcas{\var{es}}{\app{\app{\var{Cons}}{\var{e}}}{\var{es}}}{\longcas{\var{e}}{\var{Request_1}}{\app{\app{\var{Cons}}{\var{(ObsState~W~T)}}}{\brackets{\app{\var{f_2}}{\var{es}}}}}{\var{Request_2}}{\app{\app{\var{Cons}}{\var{(ObsState~T~W)}}}{\brackets{\app{\var{f_3}}{\var{es}}}}}{\var{\wildcard}}{\app{\app{\var{Cons}}{\var{(ObsState~T~T)}}}{\brackets{\app{\var{f_1}}{\var{es}}}}}}}} \\
\fundef{\var{f_2}}{\abs{\var{es}}{\smallcas{\var{es}}{\app{\app{\var{Cons}}{\var{e}}}{\var{es}}}{\longcas{\var{e}}{\var{Take_1}}{\app{\app{\var{Cons}}{\var{(ObsState~U~T)}}}{\brackets{\app{\var{f_4}}{\var{es}}}}}{\var{Request_2}}{\app{\app{\var{Cons}}{\var{(ObsState~W~W)}}}{\brackets{\app{\var{f_5}}{\var{es}}}}}{\var{\wildcard}}{\app{\app{\var{Cons}}{\var{(ObsState~W~T)}}}{\brackets{\app{\var{f_2}}{\var{es}}}}}}}} \\
\fundef{\var{f_3}}{\abs{\var{es}}{\smallcas{\var{es}}{\app{\app{\var{Cons}}{\var{e}}}{\var{es}}}{\longcas{\var{e}}{\var{Request_1}}{\app{\app{\var{Cons}}{\var{(ObsState~W~W)}}}{\brackets{\app{\var{f_5}}{\var{es}}}}}{\var{Take_2}}{\app{\app{\var{Cons}}{\var{(ObsState~T~U)}}}{\brackets{\app{\var{f_6}}{\var{es}}}}}{\var{\wildcard}}{\app{\app{\var{Cons}}{\var{(ObsState~T~W)}}}{\brackets{\app{\var{f_3}}{\var{es}}}}}}}} \\
\fundef{\var{f_4}}{\abs{\var{es}}{\smallcas{\var{es}}{\app{\app{\var{Cons}}{\var{e}}}{\var{es}}}{\cas{\var{e}}{\var{Release_1}}{\app{\app{\var{Cons}}{\var{(ObsState~T~T)}}}{\brackets{\app{\var{f_1}}{\var{es}}}}}{\var{\wildcard}}{\app{\app{\var{Cons}}{\var{(ObsState~U~T)}}}{\brackets{\app{\var{f_4}}{\var{es}}}}}}}} \\
\fundef{\var{f_5}}{\abs{\var{es}}{\smallcas{\var{es}}{\app{\app{\var{Cons}}{\var{e}}}{\var{es}}}{\longcas{\var{e}}{\var{Take_1}}{\app{\app{\var{Cons}}{\var{(ObsState~U~W)}}}{\brackets{\app{\var{f_7}}{\var{es}}}}}{\var{Take_2}}{\app{\app{\var{Cons}}{\var{(ObsState~W~U)}}}{\brackets{\app{\var{f_8}}{\var{es}}}}}{\var{\wildcard}}{\app{\app{\var{Cons}}{\var{(ObsState~W~W)}}}{\brackets{\app{\var{f_5}}{\var{es}}}}}}}} \\
\fundef{\var{f_6}}{\abs{\var{es}}{\smallcas{\var{es}}{\app{\app{\var{Cons}}{\var{e}}}{\var{es}}}{\cas{\var{e}}{\var{Release_2}}{\app{\app{\var{Cons}}{\var{(ObsState~T~T)}}}{\brackets{\app{\var{f_1}}{\var{es}}}}}{\var{\wildcard}}{\app{\app{\var{Cons}}{\var{(ObsState~T~U)}}}{\brackets{\app{\var{f_6}}{\var{es}}}}}}}} \\
\fundef{\var{f_7}}{\abs{\var{es}}{\smallcas{\var{es}}{\app{\app{\var{Cons}}{\var{e}}}{\var{es}}}{\longcas{\var{e}}{\var{Release_1}}{\app{\app{\var{Cons}}{\var{(ObsState~T~W)}}}{\brackets{\app{\var{f_3}}{\var{es}}}}}{\var{Take_2}}{\app{\app{\var{Cons}}{\var{(ObsState~U~U)}}}{\brackets{\app{\var{f_9}}{\var{es}}}}}{\var{\wildcard}}{\app{\app{\var{Cons}}{\var{(ObsState~U~W)}}}{\brackets{\app{\var{f_7}}{\var{es}}}}}}}} \\
\fundef{\var{f_8}}{\abs{\var{es}}{\smallcas{\var{es}}{\app{\app{\var{Cons}}{\var{e}}}{\var{es}}}{\longcas{\var{e}}{\var{Release_2}}{\app{\app{\var{Cons}}{\var{(ObsState~W~T)}}}{\brackets{\app{\var{f_2}}{\var{es}}}}}{\var{Take_1}}{\app{\app{\var{Cons}}{\var{(ObsState~U~U)}}}{\brackets{\app{\var{f_9}}{\var{es}}}}}{\var{\wildcard}}{\app{\app{\var{Cons}}{\var{(ObsState~W~U)}}}{\brackets{\app{\var{f_8}}{\var{es}}}}}}}} \\
\fundef{\var{f_9}}{\abs{\var{es}}{\smallcas{\var{es}}{\app{\app{\var{Cons}}{\var{e}}}{\var{es}}}{\longcas{\var{e}}{\var{Release_1}}{\app{\app{\var{Cons}}{\var{(ObsState~T~U)}}}{\brackets{\app{\var{f_6}}{\var{es}}}}}{\var{Release_2}}{\app{\app{\var{Cons}}{\var{(ObsState~U~T)}}}{\brackets{\app{\var{f_4}}{\var{es}}}}}{\var{\wildcard}}{\app{\app{\var{Cons}}{\var{(ObsState~U~U)}}}{\brackets{\app{\var{f_9}}{\var{es}}}}}}}}}}$
\end{tabular}
\caption{Example 1}
\label{example1}
\end{figure}
\begin{figure}[htbp]
\begin{pgfpicture}{0cm}{1cm}{20cm}{10cm}
\pgfnodebox{node1}[stroke]{\pgfxy(8.0,9.5)}{\parbox{1.1cm}{$f_1 \\ s_1 = T \\ s_2 = T$}}{5pt}{5pt}
\pgfnodebox{node2}[stroke]{\pgfxy(5.0,7.5)}{\parbox{1.1cm}{$f_2 \\ s_1 = W \\ s_2 = T$}}{5pt}{5pt}
\pgfsetarrowsend{stealth}
\pgfnodeconnline{node1}{node2}
\pgfnodelabel{node1}{node2}[0.5][0pt]{\pgfbox[center,center]{$Request_1$}}
\pgfnodebox{node3}[stroke]{\pgfxy(11.0,7.5)}{\parbox{1.1cm}{$f_3 \\ s_1 = T \\ s_2 = W$}}{5pt}{5pt}
\pgfnodeconnline{node1}{node3}
\pgfnodelabel{node1}{node3}[0.5][0pt]{\pgfbox[center,center]{$Request_2$}}
\pgfnodebox{node4}[stroke]{\pgfxy(2.0,5.5)}{\parbox{1.1cm}{$f_4 \\ s_1 = U \\ s_2 = T$}}{5pt}{5pt}
\pgfnodeconnline{node2}{node4}
\pgfnodelabel{node2}{node4}[0.5][0pt]{\pgfbox[center,center]{$Take_1$}}
\pgfnodeconncurve{node4}{node1}{90}{180}{2cm}{2cm}
\pgfnodelabel{node4}{node1}[0.9][1cm]{\pgfbox[center,center]{$Release_1$}}
\pgfnodebox{node5}[stroke]{\pgfxy(8.0,5.5)}{\parbox{1.1cm}{$f_5 \\ s_1 = W \\ s_2 = W$}}{5pt}{5pt}
\pgfnodeconnline{node2}{node5}
\pgfnodelabel{node2}{node5}[0.5][0pt]{\pgfbox[center,center]{$Request_2$}}
\pgfnodeconnline{node3}{node5}
\pgfnodelabel{node3}{node5}[0.5][0pt]{\pgfbox[center,center]{$Request_1$}}
\pgfnodebox{node6}[stroke]{\pgfxy(14.0,5.5)}{\parbox{1.1cm}{$f_6 \\ s_1 = T \\ s_2 = U$}}{5pt}{5pt}
\pgfnodeconnline{node3}{node6}
\pgfnodelabel{node3}{node6}[0.5][0pt]{\pgfbox[center,center]{$Take_2$}}
\pgfnodeconncurve{node6}{node1}{90}{0}{2cm}{2cm}
\pgfnodelabel{node6}{node1}[0.9][-1cm]{\pgfbox[center,center]{$Release_2$}}
\pgfnodebox{node7}[stroke]{\pgfxy(11.0,3.5)}{\parbox{1.1cm}{$f_7 \\ s_1 = U \\ s_2 = W$}}{5pt}{5pt}
\pgfnodeconnline{node5}{node7}
\pgfnodelabel{node5}{node7}[0.5][0pt]{\pgfbox[center,center]{$Take_1$}}
\pgfnodeconnline{node7}{node3}
\pgfnodelabel{node7}{node3}[0.5][0pt]{\pgfbox[center,center]{$Release_1$}}
\pgfnodebox{node8}[stroke]{\pgfxy(5.0,3.5)}{\parbox{1.1cm}{$f_8 \\ s_1 = W \\ s_2 = U$}}{5pt}{5pt}
\pgfnodeconnline{node5}{node8}
\pgfnodelabel{node5}{node8}[0.5][0pt]{\pgfbox[center,center]{$Take_2$}}
\pgfnodeconnline{node8}{node2}
\pgfnodelabel{node8}{node2}[0.5][0pt]{\pgfbox[center,center]{$Release_2$}}
\pgfnodebox{node9}[stroke]{\pgfxy(8.0,1.5)}{\parbox{1.1cm}{$f_9 \\ s_1 = U \\ s_2 = U$}}{5pt}{5pt}
\pgfnodeconnline{node7}{node9}
\pgfnodelabel{node7}{node9}[0.5][0pt]{\pgfbox[center,center]{$Take_2$}}
\pgfnodeconnline{node8}{node9}
\pgfnodelabel{node8}{node9}[0.5][0pt]{\pgfbox[center,center]{$Take_1$}}
\pgfnodeconncurve{node9}{node4}{180}{270}{2cm}{2cm}
\pgfnodelabel{node9}{node4}[0.1][0.9cm]{\pgfbox[center,center]{$Release_2$}}
\pgfnodeconncurve{node9}{node6}{0}{270}{2cm}{2cm}
\pgfnodelabel{node9}{node6}[0.1][-1cm]{\pgfbox[center,center]{$Release_1$}}
\end{pgfpicture}
\caption{LTS Representation of Example 1}
\label{example1lts}
\end{figure}

\begin{example}
\normalfont{In the second example shown in Figure \ref{example2}, each process can request access to the critical resource if 
it is thinking and the other process is not using it, take the critical resource if it is waiting for it and the other process is thinking, 
and release the critical resource if it is using it. The LTS representation of this program is shown in Figure \ref{example2lts}.}
\end{example}
\begin{figure}[htbp]
\hspace*{2cm}
\begin{tabular}{l}
$\expr{\where{\app{\app{\var{Cons}}{\var{(ObsState~T~T)}}}{\brackets{\app{\var{f_1}}{\var{es}}}}}{
\fundef{\var{f_1}}{\abs{\var{es}}{\smallcas{\var{es}}{\app{\app{\var{Cons}}{\var{e}}}{\var{es}}}{\longcas{\var{e}}{\var{Request_1}}{\app{\app{\var{Cons}}{\var{(ObsState~W~T)}}}{\brackets{\app{\var{f_2}}{\var{es}}}}}{\var{Request_2}}{\app{\app{\var{Cons}}{\var{(ObsState~T~W)}}}{\brackets{\app{\var{f_3}}{\var{es}}}}}{\var{\wildcard}}{\app{\app{\var{Cons}}{\var{(ObsState~T~T)}}}{\brackets{\app{\var{f_1}}{\var{es}}}}}}}} \\
\fundef{\var{f_2}}{\abs{\var{es}}{\smallcas{\var{es}}{\app{\app{\var{Cons}}{\var{e}}}{\var{es}}}{\longcas{\var{e}}{\var{Take_1}}{\app{\app{\var{Cons}}{\var{(ObsState~U~T)}}}{\brackets{\app{\var{f_4}}{\var{es}}}}}{\var{Request_2}}{\app{\app{\var{Cons}}{\var{(ObsState~W~W)}}}{\brackets{\app{\var{f_5}}{\var{es}}}}}{\var{\wildcard}}{\app{\app{\var{Cons}}{\var{(ObsState~W~T)}}}{\brackets{\app{\var{f_2}}{\var{es}}}}}}}} \\
\fundef{\var{f_3}}{\abs{\var{es}}{\smallcas{\var{es}}{\app{\app{\var{Cons}}{\var{e}}}{\var{es}}}{\longcas{\var{e}}{\var{Request_1}}{\app{\app{\var{Cons}}{\var{(ObsState~W~W)}}}{\brackets{\app{\var{f_5}}{\var{es}}}}}{\var{Take_2}}{\app{\app{\var{Cons}}{\var{(ObsState~T~U)}}}{\brackets{\app{\var{f_6}}{\var{es}}}}}{\var{\wildcard}}{\app{\app{\var{Cons}}{\var{(ObsState~T~W)}}}{\brackets{\app{\var{f_3}}{\var{es}}}}}}}} \\
\fundef{\var{f_4}}{\abs{\var{es}}{\smallcas{\var{es}}{\app{\app{\var{Cons}}{\var{e}}}{\var{es}}}{\cas{\var{e}}{\var{Release_1}}{\app{\app{\var{Cons}}{\var{(ObsState~T~T)}}}{\brackets{\app{\var{f_1}}{\var{es}}}}}{\var{\wildcard}}{\app{\app{\var{Cons}}{\var{(ObsState~U~T)}}}{\brackets{\app{\var{f_4}}{\var{es}}}}}}}} \\
\fundef{\var{f_5}}{\abs{\var{es}}{\smallcas{\var{es}}{\app{\app{\var{Cons}}{\var{e}}}{\var{es}}}{\smallcas{\var{e}}{\var{\wildcard}}{\app{\app{\var{Cons}}{\var{(ObsState~W~W)}}}{\brackets{\app{\var{f_5}}{\var{es}}}}}}}} \\
\fundef{\var{f_6}}{\abs{\var{es}}{\smallcas{\var{es}}{\app{\app{\var{Cons}}{\var{e}}}{\var{es}}}{\cas{\var{e}}{\var{Release_2}}{\app{\app{\var{Cons}}{\var{(ObsState~T~T)}}}{\brackets{\app{\var{f_1}}{\var{es}}}}}{\var{\wildcard}}{\app{\app{\var{Cons}}{\var{(ObsState~T~U)}}}{\brackets{\app{\var{f_6}}{\var{es}}}}}}}}}}$
\end{tabular} 
\caption{Example 2}
\label{example2}
\end{figure}
\begin{figure}[htbp]
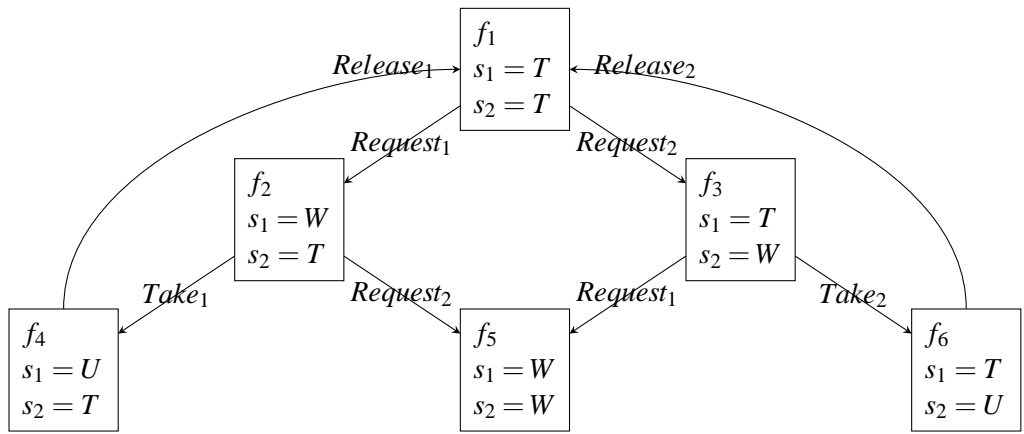

\begin{pgfpicture}{0cm}{0cm}{20cm}{5.5cm}
\pgfnodebox{node1}[stroke]{\pgfxy(8.0,4.5)}{\parbox{1.1cm}{$f_1 \\ s_1 = T \\ s_2 = T$}}{5pt}{5pt}
\pgfnodebox{node2}[stroke]{\pgfxy(5.0,2.5)}{\parbox{1.1cm}{$f_2 \\ s_1 = W \\ s_2 = T$}}{5pt}{5pt}
\pgfsetarrowsend{stealth}
\pgfnodeconnline{node1}{node2}
\pgfnodelabel{node1}{node2}[0.5][0pt]{\pgfbox[center,center]{$Request_1$}}
\pgfnodebox{node3}[stroke]{\pgfxy(11.0,2.5)}{\parbox{1.1cm}{$f_3 \\ s_1 = T \\ s_2 = W$}}{5pt}{5pt}
\pgfnodeconnline{node1}{node3}
\pgfnodelabel{node1}{node3}[0.5][0pt]{\pgfbox[center,center]{$Request_2$}}
\pgfnodebox{node4}[stroke]{\pgfxy(2.0,0.5)}{\parbox{1.1cm}{$f_4 \\ s_1 = U \\ s_2 = T$}}{5pt}{5pt}
\pgfnodeconnline{node2}{node4}
\pgfnodelabel{node2}{node4}[0.5][0pt]{\pgfbox[center,center]{$Take_1$}}
\pgfnodeconncurve{node4}{node1}{90}{180}{2cm}{2cm}
\pgfnodelabel{node4}{node1}[0.9][1cm]{\pgfbox[center,center]{$Release_1$}}
\pgfnodebox{node5}[stroke]{\pgfxy(8.0,0.5)}{\parbox{1.1cm}{$f_5 \\ s_1 = W \\ s_2 = W$}}{5pt}{5pt}
\pgfnodeconnline{node2}{node5}
\pgfnodelabel{node2}{node5}[0.5][0pt]{\pgfbox[center,center]{$Request_2$}}
\pgfnodeconnline{node3}{node5}
\pgfnodelabel{node3}{node5}[0.5][0pt]{\pgfbox[center,center]{$Request_1$}}
\pgfnodebox{node6}[stroke]{\pgfxy(14.0,0.5)}{\parbox{1.1cm}{$f_6 \\ s_1 = T \\ s_2 = U$}}{5pt}{5pt}
\pgfnodeconnline{node3}{node6}
\pgfnodelabel{node3}{node6}[0.5][0pt]{\pgfbox[center,center]{$Take_2$}}
\pgfnodeconncurve{node6}{node1}{90}{0}{2cm}{2cm}
\pgfnodelabel{node6}{node1}[0.9][-1cm]{\pgfbox[center,center]{$Release_2$}}
\end{pgfpicture}
\caption{LTS Representation of Example 2}
\label{example2lts}
\end{figure}

\begin{example}
\normalfont{In the final example shown in Figure \ref{example3}, each process can request access to the critical resource if 
it is thinking, take the critical resource if it is waiting for it and requested access before the other process, 
and release the critical resource if it is using it. Note that this program is the result of transforming an implementation of 
Lamport's bakery algorithm \cite{LAMPORT74} for two processes as shown in \cite{HAMILTON15}. Although the original 
program makes use of numbered tickets and is therefore an infinite state system, the use of tickets is completely transformed 
away and the resulting program has a finite number of states. The LTS representation of this program is shown in Figure \ref{example3lts}.}
\end{example}
\begin{figure}[htbp]
\hspace*{2cm}
\begin{tabular}{l}
$\expr{\where{\app{\app{\var{Cons}}{\var{(ObsState~T~T)}}}{\brackets{\app{\var{f_1}}{\var{es}}}}}{
\fundef{\var{f_1}}{\abs{\var{es}}{\smallcas{\var{es}}{\app{\app{\var{Cons}}{\var{e}}}{\var{es}}}{\longcas{\var{e}}{\var{Request_1}}{\app{\app{\var{Cons}}{\var{(ObsState~W~T)}}}{\brackets{\app{\var{f_2}}{\var{es}}}}}{\var{Request_2}}{\app{\app{\var{Cons}}{\var{(ObsState~T~W)}}}{\brackets{\app{\var{f_3}}{\var{es}}}}}{\var{\wildcard}}{\app{\app{\var{Cons}}{\var{(ObsState~T~T)}}}{\brackets{\app{\var{f_1}}{\var{es}}}}}}}} \\
\fundef{\var{f_2}}{\abs{\var{es}}{\smallcas{\var{es}}{\app{\app{\var{Cons}}{\var{e}}}{\var{es}}}{\longcas{\var{e}}{\var{Take_1}}{\app{\app{\var{Cons}}{\var{(ObsState~U~T)}}}{\brackets{\app{\var{f_4}}{\var{es}}}}}{\var{Request_2}}{\app{\app{\var{Cons}}{\var{(ObsState~W~W)}}}{\brackets{\app{\var{f_6}}{\var{es}}}}}{\var{\wildcard}}{\app{\app{\var{Cons}}{\var{(ObsState~W~T)}}}{\brackets{\app{\var{f_2}}{\var{es}}}}}}}} \\
\fundef{\var{f_3}}{\abs{\var{es}}{\smallcas{\var{es}}{\app{\app{\var{Cons}}{\var{e}}}{\var{es}}}{\longcas{\var{e}}{\var{Take_2}}{\app{\app{\var{Cons}}{\var{(ObsState~T~U)}}}{\brackets{\app{\var{f_5}}{\var{es}}}}}{\var{Request_1}}{\app{\app{\var{Cons}}{\var{(ObsState~W~W)}}}{\brackets{\app{\var{f_7}}{\var{es}}}}}{\var{\wildcard}}{\app{\app{\var{Cons}}{\var{(ObsState~T~W)}}}{\brackets{\app{\var{f_3}}{\var{es}}}}}}}} \\
\fundef{\var{f_4}}{\abs{\var{es}}{\smallcas{\var{es}}{\app{\app{\var{Cons}}{\var{e}}}{\var{es}}}{\longcas{\var{e}}{\var{Release_1}}{\app{\app{\var{Cons}}{\var{(ObsState~T~T)}}}{\brackets{\app{\var{f_1}}{\var{es}}}}}{\var{Request_2}}{\app{\app{\var{Cons}}{\var{(ObsState~U~W)}}}{\brackets{\app{\var{f_8}}{\var{es}}}}}{\var{\wildcard}}{\app{\app{\var{Cons}}{\var{(ObsState~U~T)}}}{\brackets{\app{\var{f_4}}{\var{es}}}}}}}} \\
\fundef{\var{f_5}}{\abs{\var{es}}{\smallcas{\var{es}}{\app{\app{\var{Cons}}{\var{e}}}{\var{es}}}{\longcas{\var{e}}{\var{Release_2}}{\app{\app{\var{Cons}}{\var{(ObsState~T~T)}}}{\brackets{\app{\var{f_1}}{\var{es}}}}}{\var{Request_1}}{\app{\app{\var{Cons}}{\var{(ObsState~W~U)}}}{\brackets{\app{\var{f_9}}{\var{es}}}}}{\var{\wildcard}}{\app{\app{\var{Cons}}{\var{(ObsState~T~U)}}}{\brackets{\app{\var{f_5}}{\var{es}}}}}}}} \\
\fundef{\var{f_6}}{\abs{\var{es}}{\smallcas{\var{es}}{\app{\app{\var{Cons}}{\var{e}}}{\var{es}}}{\cas{\var{e}}{\var{Take_1}}{\app{\app{\var{Cons}}{\var{(ObsState~U~W)}}}{\brackets{\app{\var{f_8}}{\var{es}}}}}{\var{\wildcard}}{\app{\app{\var{Cons}}{\var{(ObsState~W~W)}}}{\brackets{\app{\var{f_6}}{\var{es}}}}}}}} \\
\fundef{\var{f_7}}{\abs{\var{es}}{\smallcas{\var{es}}{\app{\app{\var{Cons}}{\var{e}}}{\var{es}}}{\cas{\var{e}}{\var{Take_2}}{\app{\app{\var{Cons}}{\var{(ObsState~W~U)}}}{\brackets{\app{\var{f_9}}{\var{es}}}}}{\var{\wildcard}}{\app{\app{\var{Cons}}{\var{(ObsState~W~W)}}}{\brackets{\app{\var{f_7}}{\var{es}}}}}}}} \\
\fundef{\var{f_8}}{\abs{\var{es}}{\smallcas{\var{es}}{\app{\app{\var{Cons}}{\var{e}}}{\var{es}}}{\cas{\var{e}}{\var{Release_1}}{\app{\app{\var{Cons}}{\var{(ObsState~T~W)}}}{\brackets{\app{\var{f_3}}{\var{es}}}}}{\var{\wildcard}}{\app{\app{\var{Cons}}{\var{(ObsState~U~W)}}}{\brackets{\app{\var{f_8}}{\var{es}}}}}}}} \\
\fundef{\var{f_9}}{\abs{\var{es}}{\smallcas{\var{es}}{\app{\app{\var{Cons}}{\var{e}}}{\var{es}}}{\cas{\var{e}}{\var{Release_2}}{\app{\app{\var{Cons}}{\var{(ObsState~W~T)}}}{\brackets{\app{\var{f_2}}{\var{es}}}}}{\var{\wildcard}}{\app{\app{\var{Cons}}{\var{(ObsState~W~U)}}}{\brackets{\app{\var{f_9}}{\var{es}}}}}}}}}}$
\end{tabular} 
\caption{Example 3}
\label{example3}
\end{figure}
\begin{figure}[htbp]
\begin{pgfpicture}{0cm}{1cm}{20cm}{10cm}
\pgfnodebox{node1}[stroke]{\pgfxy(8.0,9.5)}{\parbox{1.1cm}{$f_1 \\ s_1 = T \\ s_2 = T$}}{5pt}{5pt}
\pgfnodebox{node2}[stroke]{\pgfxy(5.0,7.5)}{\parbox{1.1cm}{$f_2 \\ s_1 = W \\ s_2 = T$}}{5pt}{5pt}
\pgfsetarrowsend{stealth}
\pgfnodeconnline{node1}{node2}
\pgfnodelabel{node1}{node2}[0.5][0pt]{\pgfbox[center,center]{$Request_1$}}
\pgfnodebox{node3}[stroke]{\pgfxy(11.0,7.5)}{\parbox{1.1cm}{$f_3 \\ s_1 = T \\ s_2 = W$}}{5pt}{5pt}
\pgfnodeconnline{node1}{node3}
\pgfnodelabel{node1}{node3}[0.5][0pt]{\pgfbox[center,center]{$Request_2$}}
\pgfnodebox{node4}[stroke]{\pgfxy(2.0,4.5)}{\parbox{1.1cm}{$f_4 \\ s_1 = U \\ s_2 = T$}}{5pt}{5pt}
\pgfnodeconnline{node2}{node4}
\pgfnodelabel{node2}{node4}[0.5][0pt]{\pgfbox[center,center]{$Take_1$}}
\pgfnodeconncurve{node4}{node1}{90}{180}{2cm}{2cm}
\pgfnodelabel{node4}{node1}[0.9][1cm]{\pgfbox[center,center]{$Release_1$}}
\pgfnodebox{node6}[stroke]{\pgfxy(5.0,4.5)}{\parbox{1.1cm}{$f_6 \\ s_1 = W \\ s_2 = W$}}{5pt}{5pt}
\pgfnodeconnline{node2}{node6}
\pgfnodelabel{node2}{node6}[0.5][0pt]{\pgfbox[center,center]{$Request_2$}}
\pgfnodebox{node7}[stroke]{\pgfxy(11.0,4.5)}{\parbox{1.1cm}{$f_7 \\ s_1 = W \\ s_2 = W$}}{5pt}{5pt}
\pgfnodeconnline{node3}{node7}
\pgfnodelabel{node3}{node7}[0.5][0pt]{\pgfbox[center,center]{$Request_1$}}
\pgfnodebox{node5}[stroke]{\pgfxy(14.0,4.5)}{\parbox{1.1cm}{$f_5 \\ s_1 = T \\ s_2 = U$}}{5pt}{5pt}
\pgfnodeconnline{node3}{node5}
\pgfnodelabel{node3}{node5}[0.5][0pt]{\pgfbox[center,center]{$Take_2$}}
\pgfnodeconncurve{node5}{node1}{90}{0}{2cm}{2cm}
\pgfnodelabel{node5}{node1}[0.9][-1cm]{\pgfbox[center,center]{$Release_2$}}
\pgfnodebox{node8}[stroke]{\pgfxy(5.0,1.5)}{\parbox{1.1cm}{$f_8 \\ s_1 = U \\ s_2 = W$}}{5pt}{5pt}
\pgfnodeconnline{node6}{node8}
\pgfnodelabel{node6}{node8}[0.5][0pt]{\pgfbox[center,center]{$Take_1$}}
\pgfnodeconnline{node4}{node8}
\pgfnodelabel{node4}{node8}[0.5][0pt]{\pgfbox[center,center]{$Request_2$}}
\pgfnodeconnline{node8}{node3}
\pgfnodelabel{node8}{node3}[0.3][0pt]{\pgfbox[center,center]{$Release_1$}}
\pgfnodebox{node9}[stroke]{\pgfxy(11.0,1.5)}{\parbox{1.1cm}{$f_9 \\ s_1 = W \\ s_2 = U$}}{5pt}{5pt}
\pgfnodeconnline{node7}{node9}
\pgfnodelabel{node7}{node9}[0.5][0pt]{\pgfbox[center,center]{$Take_2$}}
\pgfnodeconnline{node5}{node9}
\pgfnodelabel{node5}{node9}[0.5][0pt]{\pgfbox[center,center]{$Request_1$}}
\pgfnodeconnline{node9}{node2}
\pgfnodelabel{node9}{node2}[0.3][0pt]{\pgfbox[center,center]{$Release_2$}}
\end{pgfpicture}
\caption{LTS Representation of Example 3}
\label{example3lts}
\end{figure}

\section{Specification of Temporal Properties}

In this section, we describe how temporal properties of reactive systems are specified.
We use Linear-time Temporal Logic (LTL), in which the set of well-founded formulae (WFF) 
are defined inductively as follows. All atomic propositions $p$ are in WFF; if $\varphi$ and $\psi$ are in WFF, then so are:
\begin{itemize}
\item $\neg \varphi$
\item $\varphi \vee \psi$
\item $\varphi \wedge \psi$
\item $\varphi \Rightarrow \psi$
\item $\Box \varphi$
\item $\Diamond \varphi$
\item $\ocircle \varphi$
\end{itemize}
The temporal operator $\Box \varphi$ means that $\varphi$ is {\em always} true; this is used to express {\em safety}
properties. The temporal operator $\Diamond \varphi$ means that $\varphi$ will {\em eventually} be true; this is used to 
express {\em liveness} properties. The temporal operator $\ocircle \varphi$ means that $\varphi$ is true in the {\em next} state.
These modalities can be combined to obtain new modalities; for example, $\Box \Diamond \varphi$ means that $\varphi$ 
is true infinitely often, and $\Diamond \Box \varphi$ means that $\varphi$ is eventually true forever. 
Fairness constraints can also be specified for some external events (those belonging to the set $F$) which require that they occur infinitely often. For the examples given in this paper, it is assumed that all external events belong to $F$.

Here, propositional models for linear-time temporal formulas consist of a list of observable states $\pi = [s_0,s_1,\ldots]$. 
The satisfaction relation is extended to formulas in LTL for a model $\pi$ and position $i$ as follows. 
\begin{center}
\begin{tabular}{lcl}
$\pi,i \vDash p$ & iff & $p \in s_i$ \\
$\pi,i \vDash \neg \varphi$ & iff & $\pi,i \nvDash \varphi$ \\
$\pi,i \vDash \varphi \vee \psi$ & iff & $\pi,i \vDash \varphi$ or $\pi,i \vDash \psi$ \\
$\pi,i \vDash \varphi \wedge \psi$ & iff & $\pi,i \vDash \varphi$ and $\pi,i \vDash \psi$ \\
$\pi,i \vDash \varphi \Rightarrow \psi$ & iff & $\pi,i \nvDash \varphi$ or $\pi,i \vDash \psi$ \\
$\pi,i \vDash \Box \varphi$ & iff & $\forall j \geq i.\pi,j \vDash \varphi$ \\
$\pi,i \vDash \Diamond \varphi$ & iff & $\exists j \geq i.\pi,j \vDash \varphi$ \\
$\pi,i \vDash \ocircle \varphi$ & iff & $\pi,i+1 \vDash \varphi$
\end{tabular}
\end{center}
A formula $\varphi$ holds in model $\pi$ if it holds at position 0 i.e. $\pi,0 \vDash \varphi$.

The atomic propositions of these temporal formulae can be trivially translated into our functional language.
For our verification rules, we define the following datatype for truth values:
$$TruthVal ::= \mathit{True}~|~\mathit{False}~|~\mathit{Undefined}$$
We use a Kleene three-valued logic because our verification rules must always return an answer, but some of the properties to be verified may give an undefined outcome.
For our example programs which attempt to implement mutual exclusion, the following two properties are defined.
Within these temporal properties, we use the variable $s$ to denote the current observable state whose properties are being specified.
\begin{property}[Mutual Exclusion]
\normalfont{This is a safety property which specifies that both processes cannot be using the critical resource at the same time.
This can be specified as follows: \\
\hspace*{4cm} $\Box \expr{\brackets{\smallcas{\var{s}}{\var{ObsState~s_1~s_2}}{\cas{\var{s_1}}{\var{U}}{\cas{\var{s_2}}{\var{U}}{\var{False}}{\var{\wildcard}}{\var{True}}}{\var{\wildcard}}{\var{True}}}}}$
}
\end{property}
\begin{property}[Non-Starvation]
\normalfont{This is a liveness property which specifies that each process must eventually get to use the critical resource 
if they are waiting for it. This can be specified for process 1 as follows (the specification of this property for process 2 is similar): \\
\hspace*{1.5cm} $\cont{\Box (\brackets{\smallcas{\var{s}}{\var{ObsState~s_1~s_2}}{\cas{\var{s_1}}{\var{W}}{\var{True}}{\var{\wildcard}}{\var{False}}}}}
{\Rightarrow \Diamond \expr{\brackets{\smallcas{\var{s}}{\var{ObsState~s_1~s_2}}{\cas{\var{s_1}}{\var{U}}{\var{True}}{\var{\wildcard}}{\var{False}}}})}}$
}
\end{property}

\section{Verification of Temporal Properties}

In this section, we show how temporal properties of reactive systems defined in our functional language can be verified.
We define our verification rules on the restricted form of program defined in Figure \ref{simplified} as shown in Figure \ref{proofrules}.
\begin{figure}[htb]
\begin{center}
\begin{tabular}[t]{@{\hspace*{0mm}}l@{\hspace*{1mm}}l@{\hspace*{1mm}}c@{\hspace*{1mm}}l@{\hspace*{0mm}}}
(1) & $\expr{\prove{\var{e}}{(\varphi \wedge \psi)}{\phi}{\rho}}$ &  = & $(\expr{\prove{\var{e}}{\varphi}{\phi}{\rho}}) \wedge_3 (\expr{\prove{\var{e}}{\psi}{\phi}{\rho}})$ \\
(2) & $\expr{\prove{\var{e}}{(\varphi \vee \psi)}{\phi}{\rho}}$ &  = & $(\expr{\prove{\var{e}}{\varphi}{\phi}{\rho}}) \vee_3 (\expr{\prove{\var{e}}{\psi}{\phi}{\rho}})$ \\
(3) & $\expr{\prove{\var{e}}{(\varphi \Rightarrow \psi)}{\phi}{\rho}}$ &  = & $(\expr{\prove{\var{e}}{\varphi}{\phi}{\rho}}) \Rightarrow_3 (\expr{\prove{\var{e}}{\psi}{\phi}{\rho}})$ \\
(4) & $\expr{\prove{\var{e}}{(\neg \varphi)}{\phi}{\rho}}$ &  = & $\neg_3 (\expr{\prove{\var{e}}{\varphi}{\phi}{\rho}})$ \\
(5a) & $\expr{\prove{\app{\app{\var{Cons}}{\var{e_0}}}{\var{e_1}}}{(\Box \varphi)}{\phi}{\rho}}$ &  = & $(\expr{\prove{\app{\app{\var{Cons}}{\var{e_0}}}{\var{e_1}}}{\varphi}{\phi}{\emptyset}}) \wedge_3 
(\expr{\prove{\var{e_1}}{(\Box \varphi)}{\phi}{\rho}})$ \\
(5b) & $\expr{\prove{\app{\app{\var{Cons}}{\var{e_0}}}{\var{e_1}}}{(\Diamond \varphi)}{\phi}{\rho}}$ &  = & $(\expr{\prove{\app{\app{\var{Cons}}{\var{e_0}}}{\var{e_1}}}{\varphi}{\phi}{\emptyset}}) \vee_3 
(\expr{\prove{\var{e_1}}{(\Diamond \varphi)}{\phi}{\rho}})$ \\
(5c) & $\expr{\prove{\app{\app{\var{Cons}}{\var{e_0}}}{\var{e_1}}}{(\ocircle \varphi)}{\phi}{\rho}}$ &  = & $\expr{\prove{\var{e_1}}{\varphi}{\phi}{\rho}}$ \\
(5d) & $\expr{\prove{\app{\app{\var{Cons}}{\var{e_0}}}{\var{e_1}}}{\varphi}{\phi}{\rho}}$ &  = & $v$, where $\varphi[e_0/s] \Downarrow v$ \\
(6a) & $\expr{\prove{\app{\var{f}}{\args{x_1}{x_n}}}{(\Box \varphi)}{\phi}{\rho}}$ & = & $\left\{\begin{tabular}[c]{@{\hspace*{0mm}}l@{\hspace*{1mm}}l@{\hspace*{0mm}}}
$\expr{\var{True}}$, & if $f \in \rho$ \\
$\expr{\prove{\var{e}[x_1/x_1',\ldots,x_n/x_n']}{(\Box \varphi)}{\phi}{(\rho \cup \{f\})}}$, & otherwise 
\end{tabular}\right.$ \\
& & & where $\phi(f) = \expr{\abs{\args{\var{x_1'}}{\var{x_n'}}}{\var{e}}}$ \\
(6b) & $\expr{\prove{\app{\var{f}}{\args{x_1}{x_n}}}{(\Diamond \varphi)}{\phi}{\rho}}$ & = & $\left\{\begin{tabular}[c]{@{\hspace*{0mm}}l@{\hspace*{1mm}}l@{\hspace*{0mm}}}
$\expr{\var{False}}$, & if $f \in \rho$ \\
$\expr{\prove{\var{e}[x_1/x_1',\ldots,x_n/x_n']}{(\Diamond \varphi)}{\phi}{(\rho \cup \{f\})}}$, & otherwise 
\end{tabular}\right.$ \\
& & & where $\phi(f) = \expr{\abs{\args{\var{x_1'}}{\var{x_n'}}}{\var{e}}}$ \\
(6c) & $\expr{\prove{\app{\var{f}}{\args{x_1}{x_n}}}{\varphi}{\phi}{\rho}}$ & = & 
$\left\{\begin{tabular}[c]{@{\hspace*{0mm}}l@{\hspace*{1mm}}l@{\hspace*{0mm}}}
$\expr{\var{Undefined}}$, & if $f \in \rho$ \\
$\expr{\prove{\var{e}[x_1/x_1',\ldots,x_n/x_n']}{\varphi}{\phi}{(\rho \cup \{f\})}}$, & otherwise 
\end{tabular}\right.$ \\
& & & where $\phi(f) = \expr{\abs{\args{\var{x_1'}}{\var{x_n'}}}{\var{e}}}$ \\
(7a) & \multicolumn{3}{@{\hspace*{0mm}}l@{\hspace*{0mm}}}{$\expr{\prove{\casedots{\var{x}}{\var{p_1}}{\var{e_1}}{\var{p_n}}{\var{e_n}}}{(\Diamond \varphi)}{\phi}{\rho}}$} \\
& & = & $(\bigvee\limits_{p_i \in F} \expr{\prove{\var{e_i}}{(\Diamond \varphi)}{\phi}{\rho}}) \vee_3 (\bigwedge\limits_{i = 1}^{n}\expr{\prove{\var{e_i}}{(\Diamond \varphi)}{\phi}{\rho}})$ \\
(7b) & \multicolumn{3}{@{\hspace*{0mm}}l@{\hspace*{0mm}}}{$\expr{\prove{\casedots{\var{x}}{\var{p_1}}{\var{e_1}}{\var{p_n}}{\var{e_n}}}{\varphi}{\phi}{\rho}}$} \\
& & = & $\bigwedge\limits_{i=1}^{n} \expr{\prove{\var{e_i}}{\varphi}{\phi}{\rho}}$ \\
(8) & $\expr{\prove{\app{\var{x}}{\args{\var{e_1}}{\var{e_n}}}}{\varphi}{\phi}{\rho}}$ &  = & $\mathit{Undefined}$ \\
(9) & $\expr{\prove{\letexp{\var{x}}{\var{e_0}}{\var{e_1}}}{\varphi}{\phi}{\rho}}$ & = & $\expr{\prove{\var{e_1}}{\varphi}{\phi}{\rho}}$ \\
(10) & \multicolumn{3}{@{\hspace*{0mm}}l@{\hspace*{0mm}}}{$\expr{\prove{\Where{\var{e_0}}{\var{f_1}}{\var{e_1}}{\var{f_n}}{\var{e_n}}}{\varphi}{\phi}{\rho}}$} \\
& & = & $\expr{\prove{\var{e_0}}{\varphi}{(\phi \cup \{f_1 \mapsto e_1,\ldots,f_n \mapsto e_n\})}{\rho}}$
\end{tabular}
\end{center}
\caption{Verification Rules}
\label{proofrules}
\end{figure} 
The parameter $\varphi$ denotes the property to be verified and $\phi$ denotes the function variable environment.
$\rho$ denotes the set of function calls previously encountered; this is used for the detection of loops to ensure termination.
$\rho$ is also used in the verification of the $\Box$ operator (which evaluates to $True$ on encountering a loop), and the 
verification of the $\Diamond$ operator (which evaluates to $False$ on encountering a loop); $\rho$ is reset to empty when 
the verification moves inside these temporal operators. For all other temporal formulae, the value $\mathit{Undefined}$ is 
returned on encountering a loop.

The verification rules can be explained as follows. Rules (1-4) deal with the logical connectives $\wedge$, $\vee$, $\Rightarrow$ 
and $\neg$. These are implemented in our language in the usual way for a Kleene three-valued logic using the corresponding
operators $\wedge_3$, $\vee_3$, $\Rightarrow_3$ and $\neg_3$. Rules (5a-d) deal with a constructed stream of states. 
In rule (5a), if we are trying to verify that a property is always true, then we verify that it is true for the first state (with 
$\rho$ reset to empty) and is always true in all remaining states. In rule (5b), if we are trying to verify that a property is 
eventually true, then we verify that it is either true for the first state (with $\rho$ reset to empty) or is eventually true in all 
remaining states. In rule (5c), if we are trying to verify that a property is true in the next state then we verify that the property 
is true for the next state. In rule (5d), if we are trying to verify that a property is true in the current state then we verify that the 
property is true for the current state by evaluating the property using the value of the current state for the state variable $s$. 
Rules (6a-c) deal with function calls. In rule (6a), if we are trying to verify that a property is always true, then if the function 
call has been encountered before while trying to verify the same property we can return the value {\em True}; this corresponds 
to the standard greatest fixed point calculation normally used for the $\Box$ operator in which the property is initially assumed 
to be {\em True} for all states. Otherwise, the function is unfolded and added to the set of previously encountered function calls 
for this property. In rule (6b), if we are trying to verify that a property is eventually true, then if the function call has been 
encountered before while trying to verify the same property we can return the value {\em False}; this corresponds to the 
standard least fixed point calculation normally used for the $\Diamond$ property in which the property is initially assumed to be 
{\em False} for all states. Otherwise, the function is unfolded 
and added to the set of previously encountered function calls for this property. In rule (6c), if we are trying to verify that any 
other property is true, then if the function call has been encountered before we can return the value $\mathit{Undefined}$ since 
a loop has been detected. Otherwise, the function is unfolded and added to the set of previously encountered function calls. 
Rules (7a-b) deal with {\bf case} expressions. In rule (7a), if we are trying to verify that a property is eventually true,
then we verify that it is either eventually true for at least one of the branches for which there is a fairness assumption
(since these branches must be selected eventually), or that it is eventually true for all branches. 
In Rule (7b), if we are trying to verify that any other property is true, then we verify that it is true for all branches. 
In rule (8), if we encounter a free variable, then we return the value $\mathit{Undefined}$ since we cannot determine the value 
of the variable; this must be a {\bf let} variable which has been abstracted, so no information can be determined for it. 
In rule (9), in order to verify that a property is true for a {\bf let} expression, we verify that it is true for the {\bf let} body; 
this is where we perform abstraction of the extracted sub-expression. In rule (10), for a {\bf where} expression, the function 
definitions are added to the environment $\phi$.

\begin{theorem}[Soundness]
\normalfont{$\forall e \in Prog, es \in List~Event, \pi \in List~State, \varphi \in$ WFF: \\
$(e~es \overset{r*}{\leadsto} \pi) \wedge (\expr{\prove{e}{\varphi}{\emptyset}{\emptyset}} = True \Rightarrow \pi,0 \vDash \varphi) \wedge (\expr{\prove{e}{\varphi}{\emptyset}{\emptyset}} = False \Rightarrow \pi,0 \nvDash \varphi)$}
\end{theorem}
\begin{proof}
The proof of this is by structural induction on the program $e$.
\end{proof}
\begin{theorem}[Termination]
\normalfont{$\forall e \in$ Prog, $\varphi \in$ WFF: $\expr{\prove{e}{\varphi}{\emptyset}{\emptyset}}$ always terminates.}
\end{theorem}
\begin{proof}
Proof of termination is quite straightforward since there will be a finite number of functions and uses of the temporal
operators $\Box$ and $\Diamond$, and verification of each of these temporal operators will terminate when a function is
re-encountered. 
\end{proof} \\
Using these rules, we try to verify the two properties (mutual exclusion and non-starvation) for the example programs 
for mutual exclusion given in Section 3. Firstly, distillation is applied to each of the programs. 
\setcounter{example}{0}
\begin{example}
\normalfont{For the program shown in Figure \ref{example1}, Property 2 (non-starvation) holds. The verification of Property 1 (mutual 
exclusion) is shown below where we represent Property 1 by $\Box \varphi$ and the function environment by $\phi$.} \\
\\
\hspace*{0.3cm} $\expr{\prove{\app{\app{\var{Cons}}{\var{(ObsState~T~T)}}}{\brackets{\app{\var{f_1}}{\var{es}}}}}{(\Box \varphi)}{\emptyset}{\emptyset}}$ \\
= \{5a\} \\
\hspace*{0.3cm} $(\expr{\prove{\app{\app{\var{Cons}}{\var{(ObsState~T~T)}}}{\brackets{\app{\var{f_1}}{\var{es}}}}}{\varphi}{\emptyset}{\emptyset}}) \wedge_3 (\expr{\prove{\app{\var{f_1}}{\var{es}}}{(\Box \varphi)}{\emptyset}{\emptyset}})$ \\
= \{5d\} \\
\hspace*{0.3cm} $(\varphi[(ObsState~T~T)/s]) \wedge_3 (\expr{\prove{\app{\var{f_1}}{\var{es}}}{(\Box \varphi)}{\emptyset}{\emptyset}})$ \\
= \{calculation, 6a, 7b, 5a, 5d\} \\
\hspace*{0.3cm} $(\expr{\prove{\app{\var{f_1}}{\var{es}}}{(\Box \varphi)}{\phi}{\{f_1\}}}) \wedge_3 (\expr{\prove{\app{\var{f_2}}{\var{es}}}{(\Box \varphi)}{\phi}{\{f_1\}}}) \wedge_3 (\expr{\prove{\app{\var{f_3}}{\var{es}}}{(\Box \varphi)}{\phi}{\{f_1\}}})$ \\
= \{6a\} \\
\hspace*{0.3cm} $(\expr{\prove{\app{\var{f_2}}{\var{es}}}{(\Box \varphi)}{\phi}{\{f_1\}}}) \wedge_3 (\expr{\prove{\app{\var{f_3}}{\var{es}}}{(\Box \varphi)}{\phi}{\{f_1\}}})$ \\
= \{calculation, 6a, 7b, 5a, 5d\} \\
\hspace*{0.3cm} $(\expr{\prove{\app{\var{f_2}}{\var{es}}}{(\Box \varphi)}{\phi}{\{f_1,f_2\}}}) \wedge_3 (\expr{\prove{\app{\var{f_4}}{\var{es}}}{(\Box \varphi)}{\phi}{\{f_1,f_2\}}}) \wedge_3 (\expr{\prove{\app{\var{f_5}}{\var{es}}}{(\Box \varphi)}{\phi}{\{f_1,f_2\}}})$ \\
\hspace*{0.3cm} $\wedge_3 (\expr{\prove{\app{\var{f_3}}{\var{es}}}{(\Box \varphi)}{\phi}{\{f_1\}}})$ \\
= \{6a\} \\
\hspace*{0.3cm} $(\expr{\prove{\app{\var{f_4}}{\var{es}}}{(\Box \varphi)}{\phi}{\{f_1,f_2\}}}) \wedge_3 (\expr{\prove{\app{\var{f_5}}{\var{es}}}{(\Box \varphi)}{\phi}{\{f_1,f_2\}}}) \wedge_3 (\expr{\prove{\app{\var{f_3}}{\var{es}}}{(\Box \varphi)}{\phi}{\{f_1\}}})$ \\
= \{calculation, 6a, 7b, 5a, 5d\} \\
\hspace*{0.3cm} $(\expr{\prove{\app{\var{f_1}}{\var{es}}}{(\Box \varphi)}{\phi}{\{f_1,f_2,f_4\}}}) \wedge_3 (\expr{\prove{\app{\var{f_4}}{\var{es}}}{(\Box \varphi)}{\phi}{\{f_1,f_2,f_4\}}}) \wedge_3 (\expr{\prove{\app{\var{f_5}}{\var{es}}}{(\Box \varphi)}{\phi}{\{f_1,f_2\}}})$ \\
\hspace*{0.3cm} $\wedge_3 (\expr{\prove{\app{\var{f_3}}{\var{es}}}{(\Box \varphi)}{\phi}{\{f_1\}}})$ \\
= \{6a\} \\
\hspace*{0.3cm} $(\expr{\prove{\app{\var{f_5}}{\var{es}}}{(\Box \varphi)}{\phi}{\{f_1,f_2\}}}) \wedge_3 (\expr{\prove{\app{\var{f_3}}{\var{es}}}{(\Box \varphi)}{\phi}{\{f_1\}}})$ \\
= \{calculation, 6a, 7b, 5a, 5d\} \\
\hspace*{0.3cm} $(\expr{\prove{\app{\var{f_5}}{\var{es}}}{(\Box \varphi)}{\phi}{\{f_1,f_2,f_5\}}}) \wedge_3 (\expr{\prove{\app{\var{f_7}}{\var{es}}}{(\Box \varphi)}{\phi}{\{f_1,f_2,f_5\}}}) \wedge_3 (\expr{\prove{\app{\var{f_8}}{\var{es}}}{(\Box \varphi)}{\phi}{\{f_1,f_2,f_5\}}})$ \\
\hspace*{0.3cm} $\wedge_3 (\expr{\prove{\app{\var{f_3}}{\var{es}}}{(\Box \varphi)}{\phi}{\{f_1\}}})$ \\
= \{6a\} \\
\hspace*{0.3cm} $(\expr{\prove{\app{\var{f_7}}{\var{es}}}{(\Box \varphi)}{\phi}{\{f_1,f_2,f_5\}}}) \wedge_3 (\expr{\prove{\app{\var{f_8}}{\var{es}}}{(\Box \varphi)}{\phi}{\{f_1,f_2,f_5\}}}) \wedge_3 (\expr{\prove{\app{\var{f_3}}{\var{es}}}{(\Box \varphi)}{\phi}{\{f_1\}}})$ \\
= \{calculation, 6a, 7b, 5a, 5d\} \\
\hspace*{0.3cm} {\em False}
\end{example}
\begin{example}
\normalfont{For the program shown in Figure \ref{example2}, Property 1 (mutual exclusion) holds. The verification of Property 2 
(non-starvation) is shown below where we represent Property 2 by $\Box (\varphi \Rightarrow \Diamond \psi)$ and the function 
environment by $\phi$.} \\
\\
\hspace*{0.3cm} $\expr{\prove{\app{\app{\var{Cons}}{\var{(ObsState~T~T)}}}{\brackets{\app{\var{f_1}}{\var{es}}}}}{(\Box (\varphi \Rightarrow \Diamond \psi))}{\emptyset}{\emptyset}}$ \\
= \{5a\} \\
\hspace*{0.3cm} $(\expr{\prove{\app{\app{\var{Cons}}{\var{(ObsState~T~T)}}}{\brackets{\app{\var{f_1}}{\var{es}}}}}{(\varphi \Rightarrow \Diamond \psi)}{\emptyset}{\emptyset}}) \wedge_3 (\expr{\prove{\app{\var{f_1}}{\var{es}}}{(\Box (\varphi \Rightarrow \Diamond \psi))}{\emptyset}{\emptyset}})$ \\
= \{5d\} \\
\hspace*{0.3cm} $((\varphi \Rightarrow \Diamond \psi)[(ObsState~T~T)/s]) \wedge_3 (\expr{\prove{\app{\var{f_1}}{\var{es}}}{(\Box (\varphi \Rightarrow \Diamond \psi))}{\emptyset}{\emptyset}})$ \\
= \{calculation, 3, 6a, 7b, 5a, 5d\} \\
\hspace*{0.3cm} $(\expr{\prove{\app{\var{f_1}}{\var{es}}}{(\Box (\varphi \Rightarrow \Diamond \psi))}{\phi}{\{f_1\}}}) \wedge_3 (\expr{\prove{\app{\var{f_2}}{\var{es}}}{(\Box (\varphi \Rightarrow \Diamond \psi))}{\phi}{\{f_1\}}})$ \\
\hspace*{0.3cm} $\wedge_3 (\expr{\prove{\app{\var{f_3}}{\var{es}}}{(\Box (\varphi \Rightarrow \Diamond \psi))}{\phi}{\{f_1\}}})$ \\
= \{6a\} \\
\hspace*{0.3cm} $(\expr{\prove{\app{\var{f_2}}{\var{es}}}{(\Box (\varphi \Rightarrow \Diamond \psi))}{\phi}{\{f_1\}}}) \wedge_3 (\expr{\prove{\app{\var{f_3}}{\var{es}}}{(\Box (\varphi \Rightarrow \Diamond \psi))}{\phi}{\{f_1\}}})$ \\
= \{calculation, 3, 6a, 7b, 5a, 5d\} \\
\hspace*{0.3cm} $(\expr{\prove{\app{\var{f_2}}{\var{es}}}{(\Box (\varphi \Rightarrow \Diamond \psi))}{\phi}{\{f_1,f_2\}}}) \wedge_3 (\expr{\prove{\app{\var{f_4}}{\var{es}}}{(\Box (\varphi \Rightarrow \Diamond \psi))}{\phi}{\{f_1,f_2\}}})$ \\
\hspace*{0.3cm} $\wedge_3 (\expr{\prove{\app{\var{f_5}}{\var{es}}}{(\Box (\varphi \Rightarrow \Diamond \psi))}{\phi}{\{f_1,f_2\}}}) \wedge_3 (\expr{\prove{\app{\var{f_3}}{\var{es}}}{(\Box (\varphi \Rightarrow \Diamond \psi))}{\phi}{\{f_1\}}})$ \\
= \{6a\} \\
\hspace*{0.3cm} $(\expr{\prove{\app{\var{f_4}}{\var{es}}}{(\Box (\varphi \Rightarrow \Diamond \psi))}{\phi}{\{f_1,f_2\}}}) \wedge_3 (\expr{\prove{\app{\var{f_5}}{\var{es}}}{(\Box (\varphi \Rightarrow \Diamond \psi))}{\phi}{\{f_1,f_2\}}})$ \\
\hspace*{0.3cm} $\wedge_3 (\expr{\prove{\app{\var{f_3}}{\var{es}}}{(\Box (\varphi \Rightarrow \Diamond \psi))}{\phi}{\{f_1\}}})$ \\
= \{calculation, 3, 6a, 7b, 5a, 5d\} \\
\hspace*{0.3cm} $(\expr{\prove{\app{\var{f_1}}{\var{es}}}{(\Box (\varphi \Rightarrow \Diamond \psi))}{\phi}{\{f_1,f_2,f_4\}}}) \wedge_3 (\expr{\prove{\app{\var{f_4}}{\var{es}}}{(\Box (\varphi \Rightarrow \Diamond \psi))}{\phi}{\{f_1,f_2,f_4\}}})$ \\
\hspace*{0.3cm} $\wedge_3 (\expr{\prove{\app{\var{f_5}}{\var{es}}}{(\Box (\varphi \Rightarrow \Diamond \psi))}{\phi}{\{f_1,f_2\}}}) \wedge_3 (\expr{\prove{\app{\var{f_3}}{\var{es}}}{(\Box (\varphi \Rightarrow \Diamond \psi))}{\phi}{\{f_1\}}})$ \\
= \{6a\} \\
\hspace*{0.3cm} $(\expr{\prove{\app{\var{f_5}}{\var{es}}}{(\Box (\varphi \Rightarrow \Diamond \psi))}{\phi}{\{f_1,f_2\}}}) \wedge_3 (\expr{\prove{\app{\var{f_3}}{\var{es}}}{(\Box (\varphi \Rightarrow \Diamond \psi))}{\phi}{\{f_1\}}})$ \\
\ignore{
= \{6a, 7b\} \\
\hspace*{0.3cm} $(\expr{\prove{\app{\app{\var{Cons}}{\var{(ObsState~W~W)}}}{\brackets{\app{\var{f_5}}{\var{es}}}}}{(\Box (\varphi \Rightarrow \Diamond \psi))}{\emptyset}{\{f_1,f_2,f_5\}}})$ \\
\hspace*{0.3cm} $\wedge_3 (\expr{\prove{\app{\var{f_3}}{\var{es}}}{(\Box (\varphi \Rightarrow \Diamond \psi))}{\phi}{\{f_1\}}})$ \\
= \{5a\} \\
\hspace*{0.3cm} $(\expr{\prove{\app{\app{\var{Cons}}{\var{(ObsState~W~W)}}}{\brackets{\app{\var{f_5}}{\var{es}}}}}{(\varphi \Rightarrow \Diamond \psi)}{\emptyset}{\{f_1,f_2,f_5\}}})$ \\
\hspace*{0.3cm} $\wedge_3 (\expr{\prove{\app{\var{f_5}}{\var{es}}}{(\Box (\varphi \Rightarrow \Diamond \psi))}{\emptyset}{\{f_1,f_2,f_5\}}}) \wedge_3 (\expr{\prove{\app{\var{f_3}}{\var{es}}}{(\Box (\varphi \Rightarrow \Diamond \psi))}{\phi}{\{f_1\}}})$ \\
}
= \{calculation, 6a, 7b, 5a, 3, 5b, 6b\} \\
\hspace*{0.3cm} {\em False}
\end{example}
\begin{example}
\normalfont{For the program shown in Figure \ref{example3} both Property 1 (mutual exclusion) and Property 2 (non-starvation) hold.}
\end{example}
\section{Construction of Counterexamples and Witnesses}

In this section, we show how counterexamples and witnesses for temporal properties of reactive systems defined in our functional 
language can be constructed. We augment the verification rules from the previous section to generate a {\em verdict} which consists 
of a trace (a list of observable states) along with a truth value and belongs to the following datatype:
$$Verdict ::= TruthVal \times List~State$$
The trace will give a counterexample if the associated truth value is {\em False}, and a witness if the corresponding truth value is {\em True}. 
The logical connectives $\wedge_v, \vee_v, \Rightarrow_v$ and $\neg_v$ are extended to this datatype as $\wedge_v, \vee_v, 
\Rightarrow_v$ and $\neg_v$, which are defined as follows.
\begin{center}
\begin{tabular}{lcl}
$(b_1,t_1) \wedge_v (b_2,t_2)$ & = & $(b,t)$ \\
& & where \\
& & $b = b_1 \wedge_3 b_2$ \\
& & $t = min \{t_i | t_i \in \{t_1,t_2\} \wedge b_i = b\}$
\end{tabular}
\end{center}
\begin{center}
\begin{tabular}{lcl}
$(b_1,t_1) \vee_v (b_2,t_2)$ & = & $(b,t)$ \\
& & where \\
& & $b = b_1 \vee_3 b_2$ \\
& & $t = min \{t_i | t_i \in \{t_1,t_2\} \wedge b_i = b\}$
\end{tabular}
\end{center}
\begin{center}
\begin{tabular}{lcl}
$(b_1,t_1) \Rightarrow_v (b_2,t_2)$ & = & $(\neg_v (b_1,t_1)) \vee_v (b_2,t_2)$
\end{tabular}
\end{center}
\begin{center}
\begin{tabular}{lcl}
$\neg_v (b,t)$ & = & $(\neg_3 b,t)$
\end{tabular}
\end{center}
If there is more than one counterexample or witness, the function {\em min} is used to ensure that the {\em shortest} one is always returned.
The rules for the construction of counterexamples and witnesses for the simplified form of program defined in Figure \ref{simplified} are as shown in Figure \ref{counterexample}.
\begin{center}
\begin{figure}[htb]
\begin{center}
\begin{tabular}[t]{@{\hspace*{0mm}}l@{\hspace*{1mm}}l@{\hspace*{1mm}}c@{\hspace*{1mm}}l@{\hspace*{0mm}}}
(1) & $\expr{\generate{\var{e}}{(\varphi \wedge \psi)}{\phi}{\rho}{\pi}}$ &  = & $(\expr{\generate{\var{e}}{\varphi}{\phi}{\rho}{\pi}}) \wedge_v (\expr{\generate{\var{e}}{\psi}{\phi}{\rho}{\pi}})$ \\
(2) & $\expr{\generate{\var{e}}{(\varphi \vee \psi)}{\phi}{\rho}{\pi}}$ &  = & $(\expr{\generate{\var{e}}{\varphi}{\phi}{\rho}{\pi}}) \vee_v (\expr{\generate{\var{e}}{\psi}{\phi}{\rho}{\pi}})$ \\
(3) & $\expr{\generate{\var{e}}{(\varphi \Rightarrow \psi)}{\phi}{\rho}{\pi}}$ &  = & $(\expr{\generate{\var{e}}{\varphi}{\phi}{\rho}{\pi}}) \Rightarrow_v (\expr{\generate{\var{e}}{\psi}{\phi}{\rho}{\pi}})$ \\
(4) & $\expr{\generate{\var{e}}{(\neg \varphi)}{\phi}{\rho}{\pi}}$ &  = & $\neg_v (\expr{\generate{\var{e}}{\varphi}{\phi}{\rho}{\pi}})$ \\
(5a) & $\expr{\generate{\app{\app{\var{Cons}}{\var{e_0}}}{\var{e_1}}}{(\Box \varphi)}{\phi}{\rho}{\pi}}$ &  = & $(\expr{\generate{\app{\app{\var{Cons}}{\var{e_0}}}{\var{e_1}}}{\varphi}{\phi}{\emptyset}{\pi}}) \wedge_v 
(\expr{\generate{\var{e_1}}{(\Box \varphi)}{\phi}{\rho}{\var{(\concat{\pi}{[e_0]})}}})$ \\
(5b) & $\expr{\generate{\app{\app{\var{Cons}}{\var{e_0}}}{\var{e_1}}}{(\Diamond \varphi)}{\phi}{\rho}{\pi}}$ &  = & $(\expr{\generate{\app{\app{\var{Cons}}{\var{e_0}}}{\var{e_1}}}{\varphi}{\phi}{\emptyset}{\pi}}) \vee_v
(\expr{\generate{\var{e_1}}{(\Diamond \varphi)}{\phi}{\rho}{\var{(\concat{\pi}{[e_0]})}}})$ \\
(5c) & $\expr{\generate{\app{\app{\var{Cons}}{\var{e_0}}}{\var{e_1}}}{(\ocircle \varphi)}{\phi}{\rho}{\pi}}$ &  = & $\expr{\generate{\var{e_1}}{\varphi}{\phi}{\rho}{(\concat{\pi}{[e_0]})}}$ \\
(5d) & $\expr{\generate{\app{\app{\var{Cons}}{\var{e_0}}}{\var{e_1}}}{\varphi}{\phi}{\rho}{\pi}}$ &  = & $(v,\expr{\concat{\pi}{[e_0]}})$, where $\varphi[e_0/s] \Downarrow v$ \\
(6a) & $\expr{\generate{\app{\var{f}}{\args{x_1}{x_n}}}{(\Box \varphi)}{\phi}{\rho}{\pi}}$ & = & $\left\{\begin{tabular}[c]{@{\hspace*{0mm}}l@{\hspace*{1mm}}l@{\hspace*{0mm}}}
$(True,\pi)$, & if $f \in \rho$ \\
$\expr{\generate{\var{e}[x_1/x_1',\ldots,x_n/x_n']}{(\Box \varphi)}{\phi}{(\rho \cup \{f\})}{\pi}}$, & otherwise 
\end{tabular}\right.$ \\
& & & where $\phi(f) = \expr{\abs{\args{\var{x_1'}}{\var{x_n'}}}{\var{e}}}$ \\
(6b) & $\expr{\generate{\app{\var{f}}{\args{x_1}{x_n}}}{(\Diamond \varphi)}{\phi}{\rho}{\pi}}$ & = & $\left\{\begin{tabular}[c]{@{\hspace*{0mm}}l@{\hspace*{1mm}}l@{\hspace*{0mm}}}
$(False,\pi)$, & if $f \in \rho$ \\
$\expr{\generate{\var{e}[x_1/x_1',\ldots,x_n/x_n']}{(\Diamond \varphi)}{\phi}{(\rho \cup \{f\})}{\pi}}$, & otherwise 
\end{tabular}\right.$ \\
& & & where $\phi(f) = \expr{\abs{\args{\var{x_1'}}{\var{x_n'}}}{\var{e}}}$ \\
(6c) & $\expr{\generate{\app{\var{f}}{\args{x_1}{x_n}}}{\varphi}{\phi}{\rho}{\pi}}$ & = & 
$\left\{\begin{tabular}[c]{@{\hspace*{0mm}}l@{\hspace*{1mm}}l@{\hspace*{0mm}}}
$(Undefined,\pi)$, & if $f \in \rho$ \\
$\expr{\generate{\var{e}[x_1/x_1',\ldots,x_n/x_n']}{\varphi}{\phi}{(\rho \cup \{f\})}{\pi}}$, & otherwise 
\end{tabular}\right.$ \\
& & & where $\phi(f) = \expr{\abs{\args{\var{x_1'}}{\var{x_n'}}}{\var{e}}}$ \\
(7a) & \multicolumn{3}{@{\hspace*{0mm}}l@{\hspace*{0mm}}}{$\expr{\generate{\casedots{\var{x}}{\var{p_1}}{\var{e_1}}{\var{p_n}}{\var{e_n}}}{(\Diamond \varphi)}{\phi}{\rho}{\pi}}$} \\
& & = & $(\bigvee\limits_{p_i \in F} \expr{\generate{\var{e_i}}{(\Diamond \varphi)}{\phi}{\rho}{\pi}}) \vee_v (\bigwedge \limits_{i = 1}^{n}\expr{\generate{\var{e_i}}{(\Diamond \varphi)}{\phi}{\rho}{\pi}})$ \\
(7b) & \multicolumn{3}{@{\hspace*{0mm}}l@{\hspace*{0mm}}}{$\expr{\generate{\casedots{\var{x}}{\var{p_1}}{\var{e_1}}{\var{p_n}}{\var{e_n}}}{\varphi}{\phi}{\rho}{\pi}}$} \\
& & = & $\bigwedge\limits_{i=1}^{n} \expr{\generate{\var{e_i}}{\varphi}{\phi}{\rho}{\pi}}$ \\
(8) & $\expr{\generate{\app{\var{x}}{\args{\var{e_1}}{\var{e_n}}}}{\varphi}{\phi}{\rho}{\pi}}$ &  = & $(Undefined,\pi)$ \\
(9) & $\expr{\generate{\letexp{\var{x}}{\var{e_0}}{\var{e_1}}}{\varphi}{\phi}{\rho}{\pi}}$ & = & $\expr{\generate{\var{e_1}}{\varphi}{\phi}{\rho}{\pi}}$ \\
(10) & \multicolumn{3}{@{\hspace*{0mm}}l@{\hspace*{0mm}}}{$\expr{\generate{\Where{\var{e_0}}{\var{f_1}}{\var{e_1}}{\var{f_n}}{\var{e_n}}}{\varphi}{\phi}{\rho}{\pi}}$} \\
& & = & $\expr{\generate{\var{e_0}}{\varphi}{(\phi \cup \{f_1 \mapsto e_1,\ldots,f_n \mapsto e_n\})}{\rho}{\pi}}$
\end{tabular}
\end{center}
\caption{Counterexample and Witness Construction Rules}
\label{counterexample}
\end{figure} 
\end{center}
These rules are very similar to the verification rules given in Figure \ref{proofrules}, with the addition of the parameter $\pi$, which gives the
value of the current trace thus far. As each observable state in the program trace is processed in rules (5a-d),  it is appended to the end of 
$\pi$ and when a final truth value is obtained it is returned along with the value of $\pi$. Counterexamples and witnesses can of course be 
infinite in the form of a lasso consisting of a finite prefix and a loop, while only a finite trace will be returned using these rules. However,
loops can be detected in the generated trace as the repetition of observable states. To prove that the constructed counterexample or 
witness is valid, we need to prove that it satisfies the original temporal property which was verified.
\begin{theorem}[Validity]
\normalfont{$\forall e \in Prog, \varphi \in$ WFF: \\
$(\expr{\generate{e}{\varphi}{\emptyset}{\emptyset}{[]}} = (True,\pi) \Rightarrow \pi,0 \vDash \varphi) \wedge (\expr{\generate{e}{\varphi}{\emptyset}{\emptyset}{[]}} = (False,\pi) \Rightarrow \pi,0 \nvDash \varphi)$}
\end{theorem}
\begin{proof}
The proof of this is by structural induction on the program $e$.
\end{proof} \\
Using these rules, we try to construct counterexamples for the two properties (mutual exclusion and non-starvation) for the 
example programs given in Section 3. 
\setcounter{example}{0}
\begin{example}
\normalfont{For the program shown in Figure \ref{example1}, the application of these rules for Property 1 (mutual exclusion) is shown 
below where we represent Property 1 by $\Box \varphi$ and the function environment by $\phi$. We also use the shorthand notation
$(X,Y)$ to denote the state {\em ObsState} $X$ $Y$.} \\
\\
\hspace*{0.3cm} $\expr{\generate{\app{\app{\var{Cons}}{\var{(T,T)}}}{\brackets{\app{\var{f_1}}{\var{es}}}}}{(\Box \varphi)}{\emptyset}{\emptyset}{[]}}$ \\
= \{5a\} \\
\hspace*{0.3cm} $(\expr{\generate{\app{\app{\var{Cons}}{\var{(T,T)}}}{\brackets{\app{\var{f_1}}{\var{es}}}}}{\varphi}{\emptyset}{\emptyset}{[]}}) \wedge_v (\expr{\generate{\app{\var{f_1}}{\var{es}}}{(\Box \varphi)}{\emptyset}{\emptyset}{[(T,T)]}})$ \\
= \{5d\} \\
\hspace*{0.3cm} $(\varphi[(T,T)/s]) \wedge_v (\expr{\generate{\app{\var{f_1}}{\var{es}}}{(\Box \varphi)}{\emptyset}{\emptyset}{[(T,T)]}})$ \\
= \{calculation, 6a, 7b, 5a, 5d\} \\
\hspace*{0.3cm} $(\expr{\generate{\app{\var{f_1}}{\var{es}}}{(\Box \varphi)}{\phi}{\{f_1\}}{[(T,T),(T,T)]}}) \wedge_v (\expr{\generate{\app{\var{f_2}}{\var{es}}}{(\Box \varphi)}{\phi}{\{f_1\}}{[(T,T),(W,T)]}})$ \\
\hspace*{0.3cm} $\wedge_v (\expr{\generate{\app{\var{f_3}}{\var{es}}}{(\Box \varphi)}{\phi}{\{f_1\}}{[(T,T),(T,W)]}})$ \\
= \{6a\} \\
\hspace*{0.3cm} $(\expr{\generate{\app{\var{f_2}}{\var{es}}}{(\Box \varphi)}{\phi}{\{f_1\}}{[(T,T),(W,T)]}}) \wedge_v (\expr{\generate{\app{\var{f_3}}{\var{es}}}{(\Box \varphi)}{\phi}{\{f_1\}}{[(T,T),(T,W)]}})$ \\
= \{calculation, 6a, 7b, 5a, 5d\} \\
\hspace*{0.3cm} $(\expr{\generate{\app{\var{f_2}}{\var{es}}}{(\Box \varphi)}{\phi}{\{f_1,f_2\}}{[(T,T),(W,T),(W,T)]}})$ \\
\hspace*{0.3cm} $\wedge_v (\expr{\generate{\app{\var{f_4}}{\var{es}}}{(\Box \varphi)}{\phi}{\{f_1,f_2\}}{[(T,T),(W,T),(U,T)]}})$ \\
\hspace*{0.3cm} $\wedge_v (\expr{\generate{\app{\var{f_5}}{\var{es}}}{(\Box \varphi)}{\phi}{\{f_1,f_2\}}{[(T,T),(W,T),(W,W)]}})$ \\
\hspace*{0.3cm} $\wedge_v (\expr{\generate{\app{\var{f_3}}{\var{es}}}{(\Box \varphi)}{\phi}{\{f_1\}}{[(T,T),(T,W)]}})$ \\
= \{6a\} \\
\hspace*{0.3cm} $(\expr{\generate{\app{\var{f_4}}{\var{es}}}{(\Box \varphi)}{\phi}{\{f_1,f_2\}}{[(T,T),(W,T),(U,T)]}})$ \\
\hspace*{0.3cm} $\wedge_v (\expr{\generate{\app{\var{f_5}}{\var{es}}}{(\Box \varphi)}{\phi}{\{f_1,f_2\}}{[(T,T),(W,T),(W,W)]}})$ \\
\hspace*{0.3cm} $\wedge_v (\expr{\generate{\app{\var{f_3}}{\var{es}}}{(\Box \varphi)}{\phi}{\{f_1\}}{[(T,T),(T,W)]}})$ \\
= \{calculation, 6a, 7b, 5a, 5d\} \\
\hspace*{0.3cm} $(\expr{\generate{\app{\var{f_1}}{\var{es}}}{(\Box \varphi)}{\phi}{\{f_1,f_2,f_4\}}{[(T,T),(W,T),(U,T),(T,T)]}})$ \\
\hspace*{0.3cm} $\wedge_v (\expr{\generate{\app{\var{f_4}}{\var{es}}}{(\Box \varphi)}{\phi}{\{f_1,f_2,f_4\}}{[(T,T),(W,T),(U,T),(U,T)]}})$ \\
\hspace*{0.3cm} $\wedge_v (\expr{\generate{\app{\var{f_5}}{\var{es}}}{(\Box \varphi)}{\phi}{\{f_1,f_2\}}{[(T,T),(W,T),(W,W)]}})$ \\
\hspace*{0.3cm} $\wedge_v (\expr{\generate{\app{\var{f_3}}{\var{es}}}{(\Box \varphi)}{\phi}{\{f_1\}}{[(T,T),(T,W)]}})$ \\
= \{6a\} \\
\hspace*{0.3cm} $(\expr{\generate{\app{\var{f_5}}{\var{es}}}{(\Box \varphi)}{\phi}{\{f_1,f_2\}}{[(T,T),(W,T),(W,W)]}})$ \\
\hspace*{0.3cm} $\wedge_v (\expr{\generate{\app{\var{f_3}}{\var{es}}}{(\Box \varphi)}{\phi}{\{f_1\}}{[(T,T),(T,W)]}})$ \\
= \{calculation, 6a, 7b, 5a, 5d\} \\
\hspace*{0.3cm} $(\expr{\generate{\app{\var{f_5}}{\var{es}}}{(\Box \varphi)}{\phi}{\{f_1,f_2,f_5\}}{[(T,T),(W,T),(W,W),(W,W)]}})$ \\
\hspace*{0.3cm} $\wedge_v (\expr{\generate{\app{\var{f_7}}{\var{es}}}{(\Box \varphi)}{\phi}{\{f_1,f_2,f_5\}}{[(T,T),(W,T),(W,W),(U,W)]}})$ \\
\hspace*{0.3cm} $\wedge_v (\expr{\generate{\app{\var{f_8}}{\var{es}}}{(\Box \varphi)}{\phi}{\{f_1,f_2,f_5\}}{[(T,T),(W,T),(W,W),(W,U)]}})$ \\
\hspace*{0.3cm} $\wedge_v (\expr{\generate{\app{\var{f_3}}{\var{es}}}{(\Box \varphi)}{\phi}{\{f_1\}}{[(T,T),(T,W)]}})$ \\
= \{6a\} \\
\hspace*{0.3cm} $(\expr{\generate{\app{\var{f_7}}{\var{es}}}{(\Box \varphi)}{\phi}{\{f_1,f_2,f_5\}}{[(T,T),(W,T),(W,W),(U,W)]}})$ \\
\hspace*{0.3cm} $\wedge_v (\expr{\generate{\app{\var{f_8}}{\var{es}}}{(\Box \varphi)}{\phi}{\{f_1,f_2,f_5\}}{[(T,T),(W,T),(W,W),(W,U)]}})$ \\
\hspace*{0.3cm} $\wedge_v (\expr{\generate{\app{\var{f_3}}{\var{es}}}{(\Box \varphi)}{\phi}{\{f_1\}}{[(T,T),(T,W)]}})$ \\
= \{calculation, 6a, 7b, 5a, 5d\} \\
\hspace*{0.3cm} ({\em False},$[(T,T),(W,T),(W,W),(U,W),(U,U)]$) \\
\\
We can see that the rules that are applied closely mirror those applied in the verification of this property, and that the following
counterexample is generated: \\
\\
$Cons~(ObsState~T~T)~(Cons~(ObsState~W~T)~(Cons~(ObsState~W~W)~(Cons~(ObsState~U~W)$ \\
$(Cons~(ObsState~U~U)~Nil))))$
\end{example}
\begin{example}
\normalfont{For the program shown in Figure \ref{example2}, the application of these rules for Property 2 (non-starvation) is shown below 
where we represent Property 2 by $\Box (\varphi \Rightarrow \Diamond \psi)$ and the function environment by $\phi$. We again use the 
shorthand notation $(X,Y)$ to denote the state {\em ObsState} $X$ $Y$.} \\
\\
\hspace*{0.3cm} $\expr{\generate{\app{\app{\var{Cons}}{\var{(T,T)}}}{\brackets{\app{\var{f_1}}{\var{es}}}}}{(\Box (\varphi \Rightarrow \Diamond \psi))}{\emptyset}{\emptyset}{[]}}$ \\
= \{5a\} \\
\hspace*{0.3cm} $(\expr{\generate{\app{\app{\var{Cons}}{\var{(T,T)}}}{\brackets{\app{\var{f_1}}{\var{es}}}}}{(\varphi \Rightarrow \Diamond \psi)}{\emptyset}{\emptyset}{[]}}) \wedge_v (\expr{\generate{\app{\var{f_1}}{\var{es}}}{(\Box (\varphi \Rightarrow \Diamond \psi))}{\emptyset}{\emptyset}{[(T,T)]}})$ \\
= \{5d\} \\
\hspace*{0.3cm} $((\varphi \Rightarrow \Diamond \psi)[(T,T)/s]) \wedge_v (\expr{\generate{\app{\var{f_1}}{\var{es}}}{(\Box (\varphi \Rightarrow \Diamond \psi))}{\emptyset}{\emptyset}{[(T,T)]}})$ \\
= \{calculation, 3, 6a, 7b, 5a, 5d\} \\
\hspace*{0.3cm} $(\expr{\generate{\app{\var{f_1}}{\var{es}}}{(\Box (\varphi \Rightarrow \Diamond \psi))}{\phi}{\{f_1\}}{[(T,T),(T,T)]}}) \wedge_v (\expr{\generate{\app{\var{f_2}}{\var{es}}}{(\Box (\varphi \Rightarrow \Diamond \psi))}{\phi}{\{f_1\}}{[(T,T),(W,T)]}})$ \\
\hspace*{0.3cm} $\wedge_v (\expr{\generate{\app{\var{f_3}}{\var{es}}}{(\Box (\varphi \Rightarrow \Diamond \psi))}{\phi}{\{f_1\}}{\pi_4}})$ \\
= \{6a\} \\
\hspace*{0.3cm} $(\expr{\generate{\app{\var{f_2}}{\var{es}}}{(\Box (\varphi \Rightarrow \Diamond \psi))}{\phi}{\{f_1\}}{[(T,T),(W,T)]}}) \wedge_v (\expr{\generate{\app{\var{f_3}}{\var{es}}}{(\Box (\varphi \Rightarrow \Diamond \psi))}{\phi}{\{f_1\}}{[(T,T),(T,W)]}})$ \\
= \{calculation, 3, 6a, 7b, 5a, 5d\} \\
\hspace*{0.3cm} $(\expr{\generate{\app{\var{f_2}}{\var{es}}}{(\Box (\varphi \Rightarrow \Diamond \psi))}{\phi}{\{f_1,f_2\}}{[(T,T),(W,T),(W,T)]}})$ \\
\hspace*{0.3cm} $\wedge_v (\expr{\generate{\app{\var{f_4}}{\var{es}}}{(\Box (\varphi \Rightarrow \Diamond \psi))}{\phi}{\{f_1,f_2\}}{[(T,T),(W,T),(U,T)]}})$ \\
\hspace*{0.3cm} $\wedge_v (\expr{\generate{\app{\var{f_5}}{\var{es}}}{(\Box (\varphi \Rightarrow \Diamond \psi))}{\phi}{\{f_1,f_2\}}{[(T,T),(W,T),(W,W)]}})$ \\
\hspace*{0.3cm} $\wedge_v (\expr{\generate{\app{\var{f_3}}{\var{es}}}{(\Box (\varphi \Rightarrow \Diamond \psi))}{\phi}{\{f_1\}}{[(T,T),(T,W)]}})$ \\
= \{6a\} \\
\hspace*{0.3cm} $(\expr{\generate{\app{\var{f_4}}{\var{es}}}{(\Box (\varphi \Rightarrow \Diamond \psi))}{\phi}{\{f_1,f_2\}}{[(T,T),(W,T),(U,T)]}})$ \\
\hspace*{0.3cm} $\wedge_v (\expr{\generate{\app{\var{f_5}}{\var{es}}}{(\Box (\varphi \Rightarrow \Diamond \psi))}{\phi}{\{f_1,f_2\}}{[(T,T),(W,T),(W,W)]}})$ \\
\hspace*{0.3cm} $\wedge_v (\expr{\generate{\app{\var{f_3}}{\var{es}}}{(\Box (\varphi \Rightarrow \Diamond \psi))}{\phi}{\{f_1\}}{[(T,T),(T,W)]}})$ \\
= \{calculation, 3, 6a, 7b, 5a, 5d\} \\
\hspace*{0.3cm} $(\expr{\generate{\app{\var{f_1}}{\var{es}}}{(\Box (\varphi \Rightarrow \Diamond \psi))}{\phi}{\{f_1,f_2,f_4\}}{[(T,T),(W,T),(U,T),(T,T)]}})$ \\
\hspace*{0.3cm} $\wedge_v (\expr{\generate{\app{\var{f_4}}{\var{es}}}{(\Box (\varphi \Rightarrow \Diamond \psi))}{\phi}{\{f_1,f_2,f_4\}}{[(T,T),(W,T),(U,T),(U,T)]}})$ \\
\hspace*{0.3cm} $\wedge_v (\expr{\generate{\app{\var{f_5}}{\var{es}}}{(\Box (\varphi \Rightarrow \Diamond \psi))}{\phi}{\{f_1,f_2\}}{[(T,T),(W,T),(W,W)]}})$ \\
\hspace*{0.3cm} $\wedge_v (\expr{\generate{\app{\var{f_3}}{\var{es}}}{(\Box (\varphi \Rightarrow \Diamond \psi))}{\phi}{\{f_1\}}{[(T,T),(T,W)]}})$ \\
= \{6a\} \\
\hspace*{0.3cm} $(\expr{\generate{\app{\var{f_5}}{\var{es}}}{(\Box (\varphi \Rightarrow \Diamond \psi))}{\phi}{\{f_1,f_2\}}{[(T,T),(W,T),(W,W)]}})$ \\
\hspace*{0.3cm} $\wedge_v (\expr{\generate{\app{\var{f_3}}{\var{es}}}{(\Box (\varphi \Rightarrow \Diamond \psi))}{\phi}{\{f_1\}}{[(T,T),(T,W)]}})$ \\
\ignore{
= \{6a, 7b\} \\
\hspace*{0.3cm} $(\expr{\generate{\app{\app{\var{Cons}}{\var{(ObsState~W~W)}}}{\brackets{\app{\var{f_5}}{\var{es}}}}}{(\Box (\varphi \Rightarrow \Diamond \psi))}{\emptyset}{\{f_1,f_2,f_5\}}{\pi_{7}}})$ \\
\hspace*{0.3cm} $\wedge_v (\expr{\generate{\app{\var{f_3}}{\var{es}}}{(\Box (\varphi \Rightarrow \Diamond \psi))}{\phi}{\{f_1\}}{[(T,T),(T,W)]}})$ \\
= \{5a\} \\
\hspace*{0.3cm} $(\expr{\generate{\app{\app{\var{Cons}}{\var{(ObsState~W~W)}}}{\brackets{\app{\var{f_5}}{\var{es}}}}}{(\varphi \Rightarrow \Diamond \psi)}{\emptyset}{\{f_1,f_2,f_5\}}{\pi_7}})$ \\
\hspace*{0.3cm} $\wedge_v (\expr{\generate{\app{\var{f_5}}{\var{es}}}{(\Box (\varphi \Rightarrow \Diamond \psi))}{\emptyset}{\{f_1,f_2,f_5\}}{\pi_{10}}}) \wedge_v (\expr{\generate{\app{\var{f_3}}{\var{es}}}{(\Box (\varphi \Rightarrow \Diamond \psi))}{\phi}{\{f_1\}}{[(T,T),(T,W)]}})$ \\
}
= \{calculation, 6a, 7b, 5a, 3, 5b, 6b\} \\
\hspace*{0.3cm} ({\em False},$[(T,T),(W,T),(W,W),(W,W)]$) \\
\\
The following counterexample with a loop at the end is therefore generated: \\
\\
$Cons~(ObsState~T~T)~(Cons~(ObsState~W~T)~(Cons~(ObsState~W~W)~(Cons~(ObsState~W~W)~Nil)))$
\end{example}
\section{Conclusion and Related Work}

In previous work \cite{HAMILTON15}, we have shown how a fold/unfold program transformation technique can be used to verify 
both safety and liveness properties of reactive systems which have been specified using a functional language. However, counterexamples 
and witnesses were not constructed using this approach. In this paper, we have therefore extended these previous techniques to address this 
shortcoming to construct a counterexample trace when a temporal property does not hold, and a witness when it does.

Fold/unfold transformation techniques have also been developed for verifying temporal properties for logic programs 
\cite{LEUSCHEL99,ROYCHOUDHURI00,FIORAVANTI01,PETTOROSSI09,SEKI11}). Some of these techniques
have been developed only for safety properties, while others can be used to verify both safety and liveness properties.
Due to the use of a different programming paradigm, it is difficult to compare the relative power of these techniques to our own.
However, none of these techniques construct counterexamples when the temporal property does not hold.

Very few techniques have been developed for verifying temporal properties for functional programs other than the work of 
Lisitsa and Nemytykh \cite{LISITSA07,LISITSA08}. Their approach uses supercompilation \cite{TURCHIN86,SORENSEN96} as 
the fold/unfold transformation methodology, where our own approach uses distillation \cite{HAMILTON07A,HAMILTON12}.
Their approach can verify only safety properties, and does not construct counterexamples when the safety property does not hold.

One other area of work related to our own is the work on using Higher Order Recursion Schemes (HORS) to verify temporal
properties of functional programs. HORS are a kind of higher order tree grammar for generating a (potentially infinite) tree
and are well-suited to the purpose of verification since they have a decidable mu-calculus model checking problem, as proved
by Ong \cite{ONG06}. Kobayashi \cite{KOBAYASHI09} first showed how this approach can be used to verify safety properties 
of higher order functional programs and for the construction of counterexamples when the safety property does not hold. 
This approach was then extended to also verify liveness properties by Lester et al. \cite{LESTER10}, but counterexamples are
not constructed when the liveness property does not hold. These approaches have a very bad worst-case time complexity, but 
techniques have been developed to ameliorate this to a certain extent. It does however appear likely that this approach will be 
able to verify more properties than our own approach but much less efficiently.

\section*{Acknowledgements}
This work was supported, in part, by Science Foundation Ireland grant 10/CE/I1855 to Lero - the Irish Software Engineering Research Centre (www.lero.ie), and by the School of Computing, Dublin City University.

\bibliographystyle{eptcs}

\bibliography{mybib}

\end{document}